\documentclass[a4paper,10pt]{article}

\usepackage{booktabs}   

  \usepackage{graphicx}
\usepackage{subcaption}

\usepackage{amsmath}
\usepackage{mathrsfs}
\usepackage{pgf,tikz}
\usetikzlibrary{arrows}
\usepackage{a4wide}
\input{macros}
\begin{document}

\title{Probabilistic Stable Functions on Discrete Cones \\ are Power Series (long version).}

\author{Rapha\"elle Crubill\'e}

\maketitle
\begin{abstract}
  \hide{
  In ~\cite{pcsaamohpc} Danos and Ehrhard developed the category $\pcoh$ of probabilistic coherence spaces as a model of linear logic allowing to express \emph{discrete} probabilities.
  Recently~\cite{pse}, Ehrhard, Pagani and Tasson put forward the cartesian closed category $\cstabm$ of measurable cones, and measurable, stables functions between cones, as a model for continuous probabilities.
  Here, we look at $\cstabm$ as a model for \emph{discrete probabilities}, by showing the existence of a full and faithful functor embedding $\pcoh_!$ into $\cstabm$. The proof is based on a generalization of Bernstein's theorem in real analysis allowing to see stable functions between cones as generalized power series.

  }
  We study the category $\cstabm$ of measurable cones and measurable stable functions\textemdash a denotational model of an higher-order language with \emph{continuous} probabilities and full recursion~\cite{pse}. We look at $\cstabm$ as a model for \emph{discrete} probabilities, by showing the existence of a full and faithful functor preserving cartesian closed structure which embeds probabilistic coherence spaces\textemdash a \emph{fully abstract} denotational model of an higher language with full recursion and \emph{discrete} probabilities~\cite{EPT15}\textemdash into $\cstabm$. The proof is based on a generalization of Bernstein's theorem in real analysis allowing to see stable functions between discrete cones as generalized power series.
\end{abstract}

\section{Introduction}

Probabilistic reasoning allows us to describe the behavior of systems with inherent uncertainty, or on which we have an incomplete knowledge. To handle statistical models, one can employ probabilistic programming languages: they give us tools to build, evaluate and transform them.
While for some applications it is enough to consider \emph{discrete} probabilities, we sometimes want to model systems where the underlying space of events has inherent \emph{continuous} aspects: for instance in hybrid control systems~\cite{alur1996hybrid}, as used e.g. in flight management.
In the machine learning community~\cite{gordon2014probabilistic,goodman2013principles}, statistical models are also used to express our \emph{beliefs} about the world, that we may then update using \emph{Bayesian inference}\textemdash the ability to condition values of variables via observations.

As a consequence, several probabilistic \emph{continuous} languages have been introduced and studied, such as Church~\cite{church}, Anglican~\cite{anglican}, as well as formal operational semantics for them~\cite{borgstrom2016lambda}. Giving a \emph{fully abstract} \emph{denotational} semantics to a higher-order probabilistic language with full recursion, however, has proved to be harder than in the non-probabilistic case. For discrete probabilities, there have been two such fully abstract models: in~\cite{danos2002probabilistic}, Danos and Harmer introduced a fully abstract denotational semantics of a probabilistic extension of idealized Algol, based on game semantics; and in~\cite{pcsaamohpc} Ehrhard, Pagani and Tasson showed that the category $\pcoh$ of \emph{probabilistic coherence spaces} gives a fully abstract model for $\pcf_\oplus$, a discrete probabilistic variant of Plotkin's $\pcf$.

While there is currently no known fully abstract denotational semantics for a higher-order language with full recursion and continuous probabilities, several denotational models have been introduced. The pioneering work of Kozen~\cite{kozen1979semantics} gave a denotational semantics to a \emph{first-order} while-language endowed with a random real number generator. In~\cite{staton2016semantics}, Staton et al
give a denotational semantics to an higher-order language: they first develop a distributive category based on measurable spaces as a model of the first-order fragment of their language, and then extend it into a cartesian closed category using a standard construction based on the functor category.

Recently, Ehrhard, Pagani and Tasson introduced in~\cite{pse} the category $\cstabm$, as a denotational model of an extension of $\pcf$ with continuous probabilities. It is presented as a refinement with \emph{measurability constraints} of the category $\cstab$ of abstract cones and so-called \emph{stable} functions between cones, consisting in a generalization of \emph{absolutely monotonous functions} from real analysis.

Here, we look at the category $\cstabm$ from the point of view of \emph{discrete} probabilities. It was noted in~\cite{pse} that there is a natural way to see any probabilistic coherent space as an object of $\cstab$. In this work, we show that this connection leads to a full and faithful functor $\functor$ from $\pcoh_!$\textemdash the Kleisli category of $\pcoh$\textemdash into $\cstab$. It is done by showing that every stable function between probabilistic coherent spaces can be seen as a power series, using the extension to an abstract setting of Bernstein's theorem for absolutely monotonous functions shown by McMillan~\cite{mcmillan}.  We then show that the functor $\functor$ we have built is cartesian closed, i.e. respects the cartesian closed structure of $\pcoh_!$. In the last part, we turn $\functor$ into a functor $\functorm:\pcoh_! \rightarrow \cstabm$, and we show that $\functorm$ too is cartesian closed.

To sum up, the contribution of this paper is to show that there is a cartesian closed full embedding from $\pcoh_!$ into $\cstabm$. Since $\pcoh_!$ is known to be a fully abstract denotational model of $\pcf_\oplus$, a corollary of this result is that $\cstabm$ too is a fully abstract model of $\pcf_\oplus$.

\section{Discrete and Continuous Probabilistic Extension of $\pcf$: an Overview.}\label{sect:overview}
A simple way to add probabilities to a (higher-order) programming language is to add a fair probabilistic choice operator to the syntax.
Such an approach has been applied to various extensions of the $\lambda$-calculus~\cite{DLZ}. To fix ideas, we give here the syntax of a (minimal) probabilistic variant of Plotkin's $\pcf$~\cite{plotkin1977lcf}, that we will call $\pcf_\oplus$. It is a typed language, whose types are given by: $A \bnf N \midd A \rightarrow A$, where $N$ is the base type of naturals numbers. The programs are generated as follows:
\begin{align*}
  \termone, & \termtwo \in \pcf_\oplus \bnf\; \varone \midd \lambda {\varone^A}\cdot \termone \midd (\termone \termtwo) \midd (Y \termtwo) \\ &
\midd \text{ifz }(\termone, \termtwo, \termthree) \midd \text{let}(\varone, \termone, \termtwo)\midd \termone \oplus \termtwo \\ &  \midd \underline{n} \midd \text{ succ }(\termone) \midd \text{ pred }(\termone)
\end{align*}
The operator $\oplus$ is the fair probabilistic choice operator, $Y$ is a recursion operator, and $n$ ranges over natural numbers. The $\text{ifz}$ construct tests if its first argument (of type $N$) is $0$, reduces to its second argument if it is the case, and to its third otherwise. We endow this language with a natural operational semantics~\cite{EPT15}, that we choose to be call-by-name. However, for expressiveness we need to be able to simulate a call-by-value discipline on terms of ground type $N$: it is enabled by the $\text{let}$-construct.  

We can see that the kind of probabilistic behavior captured by $\pcf_\oplus$ is \emph{discrete}, in the sense that it manipulates distributions on countable sets.
In~\cite{pcsaamohpc}, Ehrhard and Danos introduced a model of Linear Logic designed to lead to denotational models for discrete higher-order probabilistic computation: the category $\pcoh$ of \emph{probabilistic coherence spaces (PCSs)}. It was indeed shown in~\cite{EPT15} that $\pcoh_!$, the Kleisli category of $\pcoh$ is a \emph{fully abstract} model of $\pcf_\oplus$, while the Eilenberg-Moore Category of $\pcoh$ is a \emph{fully abstract} model of a probabilistic variant of Levy's call-by-push-value calculus.

We are going to illustrate here on examples the ideas behind the denotational semantics of $\pcf_\oplus$ in $\pcoh_!$.
The basic idea is that the denotation of a program consists of a vector on $\Rp{X}$, where $X$ is the countable sets of possible outcomes. For instance, the denotation of the program $\underline{0} \oplus \underline{1}$ of type $N$ is the vector $x \in \Rp{\mathbb{N}}$, with $x_0 = \frac 1 2$, $x_1 = \frac 1 2$, and $x_k = 0$ for $k \not \in \{0,1\}$. Morphisms in $\pcoh_!$, on the other hand, can be seen as analytic functions (i.e. power series) between real vector spaces. Let us look at the denotation of the simple $\pcf_\oplus$ program below.
$$M := \lambda {\varone^N}\cdot \left(\underline{0} \oplus \text{ifz}(x,\underline{1}, \text{ifz}(x,\underline{0},\Omega)) \right),$$
where $\Omega$ is the usual encoding of a never terminating term using the recursion operator.
The denotation of $M$ consists of the following function $\Rp{\mathbb{N}} \rightarrow \Rp{\mathbb{N}}$:
$$f(x)_k = \begin{cases}
  \frac 1 2 + \frac 1 2 \sum_{i \neq 0} x_i \cdot x_0 \qquad \text{ if } k=0 \\
  \frac 1 2 x_0 \qquad \text{ if }k = 1 \\
  0 \qquad \text{ if }k \not \in \{0,1\}
\end{cases}
$$
We can see that $f(x)_k$ corresponds indeed to the probability of obtaining $\underline{k}$ if we pass to $M$ a term $N$ with $x$ as denotation. Observe that $f$ here is a polynomial in $x$; however since we have recursion in our language, there are programs that do an unbounded number of calls to their arguments: then their denotations are not polynomials anymore, but they are still analytic functions. The analytic nature of $\pcoh_!$ morphisms plays a key role in the proof of full abstraction for $\pcf_\oplus$.

Observe that this way of building a model for $\pcf_\oplus$ is utterly dependent on the fact that we consider discrete probabilities over a countable sets of values. In recent years, however, there has been much focus on continuous probabilities in higher-order languages. The aim is to be able to handle classical mathematical distributions on reals, as for instance normal or Gaussian distributions, that are widely used to build generic physical or statistical models, as well as transformations over these distributions.

We illustrate the basic idea here by presenting the language $\pcf_{\texttt{sample}}$, following~\cite{pse}, that can be seen as the continuous counterpart to the discrete language $\pcf_\oplus$. It is a typed language, with types generated as $A \bnf R \midd A \rightarrow A$, and terms generated as follows:
\begin{align*}
  \termone \in \pcf_{\text{sample}} \bnf\;& \varone \midd \lambda {\varone^A}\cdot \termone \midd (\termone \termtwo) \midd (Y \termtwo) \\ &
\midd \text{ifz }(\termone, \termtwo, \termthree) \midd \text{let}(\varone, \termone, \termtwo) \\
& \midd \underline{r} \midd \texttt{sample} \midd \underline{f}(\termone_1, \ldots, \termone_n) 
 \end{align*}
where $r$ is any real number, and $f$ is in a fixed countable set of measurable functions $\mathbb{R}^n \rightarrow \mathbb{R}$. The constant \texttt{sample} stands for the uniform distribution over $[0,1]$.
Observe that admitting every measurable functions as primitive in the language allows to encode every distribution that can be obtained in a measurable way from the uniform distribution, for instance Gaussian or normal distributions. This language is actually expressive enough to simulate other probabilistic features, as for instance Bayesian conditioning, as highlighted in~\cite{pse}.
Moreover, we can argue it is also \emph{more general} than $\pcf_\oplus$: first it allows to encode integers (since $\mathbb{N} \subseteq \mathbb{R}$) and basic arithmetic operations over them. Secondly, since the orders operator $ \geq : \mathbb{R} \times \mathbb{R} \rightarrow \{0,1\} \subseteq \mathbb{R}$ is measurable, we can construct in $\pcf_{\text{sample}}$ terms like this one:
$$\text{ifz}(\underline{\geq}(\texttt{sample}, \frac 1 2), \termone, \termtwo),$$
which encodes a fair choice between $\termone$ and $\termtwo$.

We see, however, that $\pcoh_!$ cannot be a model for $\pcf_{\texttt{sample}}$: indeed it doesn't even seem possible to write a probabilistic coherence space for the real type.
In~\cite{pse}, Ehrhard, Pagani and Tasson introduced the cartesian closed category $\cstabm$ of measurable cones and measurable stables functions, and showed that it provides an adequate and sound denotational model for $\pcf_{\texttt{sample}}$. The denotation of the base type $R$ is taken as the set of finite measures over reals, and the denotation of higher-order types is then built naturally using the cartesian closed structure. From there, it is natural to ask ourselves: how \emph{good} $\cstabm$ is as a model of probabilistic higher-order languages ?

The present paper is devoted to give a partial answer to this question: in the case where we restrict ourselves to a \emph{discrete fragment} of $\pcf_{\texttt{sample}}$. To make more precise what we mean, let us consider a continuous language with an \emph{explicit} discrete fragment which has both $R$ and $N$ as base types: we consider the language $\pcf_{\oplus,\texttt{sample}}$ with all syntactic constructs of both $\pcf_{\oplus}$ and $\pcf_{\texttt{sample}}$, as well as an operator $\texttt{real}$ with the typing rule:
$$
\AxiomC{$\strut \wfjt {\contone} \termone N $}
\UnaryInfC{$\wfjt {\contone} {\mathtt{real}(\termone)} R  $}
\DisplayProof,
$$
designed to enable the continuous constructs to act on the discrete fragment, by giving a way to see any distribution on $\NN$ as a distribution on $\RR$.  
We see that we can indeed extend in a natural way the denotational semantics of $\pcf_{\texttt{sample}}$ given in~\cite{pse} to $\pcf_{\texttt{sample}, \oplus}$: in the same way that the denotational semantics of $R$ is taken as the set of all finite measures on $\RR$, we take the denotational semantics of $N$ as the set $\meas \NN$ of all finite measures over $\NN$. We take as denotational semantics of the operator $\texttt{real}$ the function:
$$\semm {\texttt{real}}: \mu \in \meas \NN \mapsto \left(A \in \Sigma_\RR \mapsto \sum_{n \in \NN \cap A} \mu(n) \right) \in \meas {\RR}.$$
We will see later that this function is indeed a morphism in $\cstabm \allowbreak(\meas{\NN}, \meas {\RR})$. What we would like to know is: what is the structure of the sub-category of $\cstabm$ given by the \emph{discrete types of $\pcf_{\texttt{sample, }\oplus}$}, i.e the one generated inductively by $\semm N$, $\Rightarrow$, $\times$ ?

The starting point of our work is the connection highlighted in~\cite{pse} between PCSs and complete cones: every PCSs can be seen as a complete cone, in such a way that the denotational semantics of $N$ in $\pcoh_!$ becomes the set of finite measures over $\NN$.  We formalize this connection by a functor $F^m : \pcoh_! \rightarrow \cstabm$. 
However, to be able to use $\pcoh_!$ to obtain information about the discrete types sub-category of $\cstabm$, we need to know whether this connection is preserved at higher-order types: does the $\Rightarrow$ construct in $\cstabm$ make some wild functions not representable in $\pcoh_!$ to appear, e.g. not analytic? The main technical part of this paper consists in showing that this is not the case, meaning that the functor $F^m$ is full and faithful, and cartesian closed. It tells us that the discrete types sub-category of $\cstabm$ has actually the same structure as the subcategory of $\pcoh_!$ generated by $\sem \NN^\pcoh$, $\Rightarrow$ and $\times$. Since $\pcoh_!$ is a fully abstract model of $\pcf_\oplus$, it tells us that the discrete fragment of $\pcf_{\texttt{sample},\oplus}$ is fully abstract in $\cstabm$.

\section{Cones and Stable Functions}
The category of measurable cones and measurable, stable functions ($\cstabm$), was introduced by Ehrhard, Pagani, Tasson in~\cite{pse} in the aim to give a model for $\pcf_{\texttt{sample}}$.

They actually introduced it as a refinement of the category of complete cones and stable functions, denoted $\cstab$.
 Stable functions on cones are a generalization of well-known \emph{absolutely monotonic functions} in real analysis: they are those functions $f: [0, \infty) \rightarrow \Rp{}$ which are infinitely differentiable, and such that moreover all their derivatives are non-negative. The relevance of such functions comes from a result due to Bernstein: every absolutely monotonic function coincides with a power series. 
Moreover, it is possible to characterize absolutely monotonic functions without explicitly asking for them to be differentiable: it is exactly those functions such that all the so-called \emph{higher-order differences}, which are quantities defined only by sum and subtraction of terms of the form $f(x)$, are non-negative. (see~\cite{widder}, chapter 4). The definition of \emph{pre-stable functions} in~\cite{pse} generalizes this characterization.  

In this section, we first recall basic facts about cones and stable functions, all extracted from ~\cite{pse}. Then we will prove a generalization of Bernstein's theorem for pre-stable functions over a particular class of cones, which is the main technical contribution of this paper. We will do that following the work of McMillan on a generalization of Bernstein's theorem for functions ranging over abstract domains endowed with partition systems, see~\cite{mcmillan}.




\subsection{Cones}
The use of a notion of cones in denotational semantics to deal with probabilistic behavior goes back to Kozen in~\cite{kozen1979semantics}.
 We take here the same definition of cone as in~\cite{pse}.

\begin{definition}
  A cone $\coneone$ is a $\Rp{}$-semimodule given together with an $\Rp{}$ valued function $\norm \coneone {\cdot}$ called \emph{norm of $\coneone$}, and verifying:
  \begin{align*}
    & \left(x + y = x + y'\right)\,  \Rightarrow\, y = y' && \norm \coneone {\alpha x} = \alpha {\norm \coneone x} \\
    & \norm {\coneone}{x + x'} \leq \norm \coneone x + \norm \coneone {x'} &&
    \norm \coneone x = 0 \Rightarrow x = 0\\
    & \norm \coneone x \leq \norm \coneone {x + x'}
    \end{align*}
\end{definition}

The most immediate example of cone is the non-negative real half-line, when we take as norm the identity. Another example is the positive quadrant in a 2-dimensional plan, endowed with the euclidian norm. In a way, the notion of cones is the generalization of the idea of a space where all elements are \emph{non-negative}.
This analogy gives us a generic way to define a pre-order, using the $+$ of the cone structure.
\begin{definition}\label{def:coneorder}
Let be $\coneone$ a cone. Then we define a partial order $\order \coneone$ on $\coneone$ by: $x {\order \coneone} y $ if there exists $z \in \coneone$, with $y = x + z$.
\end{definition}
We define $\boule \coneone$ as the set of elements in $\coneone$ of norm smaller or equal to $1$. We will sometimes call it the \emph{unit ball} of $\coneone$. Moreover, we will also be interested in the \emph{open unit ball} $\bouleopen \coneone$, defined as the set of elements of $\coneone$ of norm smaller than $1$.

In~\cite{pse}, the authors restrict themselves to cones verifying a completeness criterion: it allows them to define the denotation of the recursion operator in $\pcf_{\text{sample}}$, thus enforcing the existence of fixpoints.
\begin{definition}\label{def:cc}
  A cone $\coneone$ is said to be:
  \begin{itemize}
  \item \emph{sequentially complete} if any non-decreasing sequence $(x_n)_{n \in \NN}$ of elements of $\boule \coneone$ has a least upper bound $\sup_{n \in \NN} x_n \in \boule \coneone$.
  \item \emph{directed complete} if for any directed subset $D$ of $\boule \coneone$, $D$ has a least upper bound $\sup D \in\boule \coneone$.
    \item \emph{a lattice cone} if any two elements $x,y$ of $\coneone$ have a least upper bound $x \vee y $.
    \end{itemize}
\end{definition}
Observe that a directed-complete cone is always sequentially complete.
\longv{
  \begin{lemma}\label{lemma:lattice}
    Let be $\coneone$ a lattice cone. Then it holds that:
    \begin{itemize}
    \item Any two element $x,y$ of $\coneone$ have a greatest lower bound $x \wedge y$.
      \item Decomposition Property: if $z \leq x +y$, there there exists $z_1, z_2 \in \coneone$ such that $z = z_1 + z_2$, and $z_1 \leq x$, and $z_2 \leq y$.
      \end{itemize}
  \end{lemma}
  \begin{proof}
    Recall that, if $a \geq b$, we denote by $a-b$ the element $c$ such that $a = b+c$.
    \begin{itemize}
    \item We consider $z = x+y -(x  \vee y)$, and we show that $z$ is indeed the greatest lower bound of $x$ and $y$.
      \item We take $z_2 = (x \vee z)-x$, and $z_1 = z - z_1$.  First, we see that $z_2 \leq (x+y) - x$, and so $z_2 \leq y$. Moreover, $z_1 = x -( (x \vee z)-z) \leq x$.
    \end{itemize}
    \end{proof}
  }
We illustrate Definition~\ref{def:cc} by giving the complete cone used in~\cite{pse} as the denotational semantics of the base type $R$ in $\pcf_{\text{sample}}$.
\begin{example}\label{ex:R}
  We take $\meas \RR$ as the set of \emph{finite} measures over $\mathbb{R}$, and the norm as $\norm {\meas {\RR}} \mu = \mu(\mathbb{R})$. $\meas \RR$ is a directed-complete cone. For every $r \in \mathbb{R}$, the denotational semantics of the term $\underline{r}$ in~\cite{pse} is $\delta_r$, the \emph{Dirac measure with respect to $r$} defined by taking
  $\delta_r(U) = 1 \text{ if } r \in U$, and $\delta_r(U) = 0$ otherwise. 
\end{example}
\hide{
  \begin{proof}
  Let be $D$ a directed subset of $\boule {C_R}$. Recall that the elements of $\boule{C_R}$ are the finite measures $\nu$ on $\mathbb{R}$ such that $\nu(\RR^n) \leq 1$. We define a function $\mu : \Sigma_\RR \rightarrow \Rp{}$, as $\mu(X) = \sup_{\nu \in X} \nu(X)$.
  We are going to show that $\mu$ is a finite measure on $\mathbb{R}$, and that moreover it is the lub of $D$.
  We first show that $\mu$ is a measure: it is direct to see that it is non-negative, and that $\mu(\emptyset) = 0$. We show now the countable additivity: let be $(X_n)_{n \in \NN}$ a sequence of disjoints measurable subsets of $\mathbb{R}$. First, we see that:
    \begin{align*}
      \mu(\sum_{n \in \NN} X_n) & = \sup_{\nu \in D} \nu(\sum_{n \in \NN} X_n) \\
      & =  \sup_{\nu \in D} \sum_{n \in \NN} \nu( X_n)  \text{ since every }\nu\text{ is a measure}\\
      & \leq \sup_{\nu \in D} \sum_{n \in \NN} \mu( X_n)  = \sum_{n \in \NN} \mu( X_n).
    \end{align*}
    We have still to show the reverse inequality. Let be $\nu_0 \in D$.
    Let be $\epsilon > 0$. Since $\nu_0$ is a \emph{finite} measure, there exists $N \in \NN$, such that $\nu_0(\sum_{N \leq n } X_n) \leq \epsilon$. It means that for every $\nu \in D$ with $\nu \leq \nu_0$, $\nu(\sum_{N \leq n } X_n) \leq \epsilon$.
    As a consequence:
    \begin{align*}
      \mu(\sum_{n \in \NN} X_n) & = \inf_{\nu \in D \text{ with }\nu \leq \nu_0} \nu(\sum_{n \in \NN} X_n) \\
      & \leq \inf_{\nu \in D \text{ with }\nu \leq \nu_0} (\sum_{n \leq N} \nu( X_n) + \epsilon)\\  
      & =  \sum_{n \leq N} \mu( X_n)  + \epsilon \leq \sum_{n \in \NN} \mu(X_n) + \epsilon
    \end{align*}
    Since it is true for every $\epsilon$, it means that $\mu(\sum_{n \in \NN} X_n) \leq \sum_{n \in \NN} \mu(X_n)$. Since we have already show the other inequality, it holds that $\mu$ is a measure. Looking at the definition of $\mu$, it is immediate that it is finite, and the greatest lower bound of $D$.
  \end{proof}}
In a similar way, we define $\meas X$ as the directed-complete cone of finite measures over $X$, for any measurable space $X$.

In~\cite{pse}, the authors ask for the cones they consider only to be sequentially complete. It is due to the fact they want to add measurability requirements to their cones, and as a rule, sequential completeness interacts better with measurability than directed completeness since measurable sets are closed under \emph{countable} unions, but not \emph{general unions}.
\longv{
We illustrate this point in the example below.
\begin{example}
  Let be $A$ a measurable space, and $\mu$ a finite measure on $A$.
We consider the cone of measurable functions $A \rightarrow \RR_+$. We take $\norm{}f = \int_{A} f d\mu$. Lebesgues Monotone Convergence Theorem shows that this cone is sequentially complete, but it is not directed complete.
\end{example}}
In this work however, we are only interested in cones arising from \emph{probabilistic coherence spaces} in a way we will develop in Section~\ref{sect:pcs}. Since those cones have an underlying \emph{discrete} structure, we will be able to show that they are actually \emph{directed complete}. We will need this information, since we will apply McMillan's results~\cite{mcmillan} obtained in the more general framework of abstract domains with partitions, in which he asks for directed completeness. That's because
directed completness allows to also enforce the existence of \emph{infinum}, as stated in the lemma below, whose proof can be found in the long version.
\begin{lemma}
  If a cone $\coneone$ is:
  \begin{itemize}
    \item \emph{sequentially} complete, then  every non-increasing sequence $(x_n)_{n \in \NN}$ has a greatest lower bound $\inf (x_n)_{n \in \NN}$.
  \item \emph{directed} complete, then for every $D \subseteq \coneone$ directed for the reverse order, $D$ has a greatest lower bound $\inf D$.
    \end{itemize}
\end{lemma}
\longv{\begin{proof}
We do the proof when $\coneone$ is \emph{directed} complete, but it is exactly the same in the sequentially complete case.
    Let be $D$ a reverse directed set. If all elements of $D$ are zero, then $\inf D = 0$. Otherwise, let be $x>0 \in D$. We consider the subset $E = \{\frac {x - y}{\norm{\coneone} x} \mid y \leq x \wedge y \in D \}$. It is easy to see it is a directed subset of $\boule \coneone$, which means that, since $\coneone$ is directed complete, it has a supremum. So we can take $z = \norm {\coneone}x \cdot{(x - \sup E)}$, and we show that it is the least upper bound of $D$.
    \end{proof}
}
It is shown in~\cite{pse} that the addition and multiplication by a scalar are Scott-continuous in complete cones, in a sequential sense. In directed complete cones, it holds also in a \emph{directed sense}.
\begin{lemma}
  The addition $+ : \coneone \times \coneone \rightarrow \coneone$ and the scalar multiplication $\cdot : \Rp{} \times \coneone \rightarrow \coneone$ are Scott-continuous:
  \begin{itemize}
  \item for any directed subsets $D$ and $E$ of $\coneone$, and $K$ of $\Rp{}$:
    \begin{align*}
      &\sup{\{x + y \mid x \in D, y \in E\}} = \sup D + \sup E;\\
      \text{and }\quad & \sup\{\lambda \cdot x \mid \lambda \in K, \, x \in E  \} = \sup K \cdot \sup E.
      \end{align*}
  \item for any reverse directed subsets $D$, $E$ of $\coneone$, and $K$ of $\Rp{}$:
\begin{align*}
  &\inf{\{x + y \mid x \in D, y \in E\}} = \inf D + \inf E; \\
  \text{and }\quad &\inf\{\lambda \cdot x \mid \lambda \in K, \, x \in E  \} = \inf K \cdot \inf E.
  \end{align*}
    \end{itemize}
\end{lemma}


\subsection{Pre-Stable Functions between Cones}\label{subsect:cstab}

As said before, the notion of pre-stable function is a generalization of the notion of \emph{absolutely monotonic} real functions. More precisely, the idea is to define so-called \emph{higher-order differences}, and to specify that they must be all non-negative.

First, we want to be able to talk about those $\vec u = (u_1, \ldots ,u_n)$, such that $\norm{\coneone}{x + \sum u_i} \leq 1$ for a fixed $x \in \boule \coneone$, and $n \in \NN$. To that end, we introduce a cone $\coneone_x^n$ whose unit ball is exactly such elements. It is an adaptation of the definition given in~\cite{pse} for the case where $n=1$, and we show in the same way that it is indeed a cone.

\begin{definition}[Local Cone]
  Let be $\coneone$ a cone, $n \in \NN$, and $x \in \bouleopen \coneone$. We call \emph{$n$-local cone at $x$}, and we denote $\coneone_x^n$ the cone $\coneone^n$ endowed with the following norm:
  $$\norm {\coneone_x^n} {(u_1, \ldots, u_n)} =  \inf {\{ \frac 1 r \mid  x + r \cdot \sum_{1 \leq i \leq n} u_i \in \boule \coneone \wedge r>0 \}}.$$
  \end{definition}
We can show that whenever $\coneone$ is a directed-complete cone, $\coneone_x^n$ is also directed-complete.
\remarque{est ce qu'on en a besoin ?}

For $n\in \NN$, we use $\parteps{+} {n}$ (respectively $\parteps{-}n$) for the set of all subsets $I$ of $\{1,\ldots, n\}$ such that $n - {\card I}$ is even (respectively odd).

We are now ready to introduce \emph{higher-order differences}. Since we have only explicit addition, not subtraction, we define separately the positive part $\diff + n$ and the negative part $\diff - n$ of those differences:
For $f: \boule \coneone \rightarrow \conetwo$, $x \in \boule \coneone$, $\vec u \in \boule \coneone_x^n$, and $\epsilon \in \{-,+\}$, we define:
  \begin{align*}
\diff \epsilon n (f)(x \mid \vec u) &= \sum_{I \in \parteps \epsilon n} f(x + \sum_{i \in I } u_i)
  \end{align*}

\begin{definition}
  We say that $f: \boule \coneone \rightarrow \conetwo$ is \emph{pre-stable} if, for every $n \in \NN$, for every   $x \in \boule \coneone$, $\vec u \in \boule \coneone_x^n$, it holds that:
  $$\diff - n (f)(x \mid \vec u) \leq \diff + n (f)(x \mid \vec u).$$
  \end{definition}
If $f$ is pre-stable, we will set $\diff{}n f (x \mid \vec u) = \diff + n f (x \mid \vec u) - \diff - n f (x \mid \vec u)$. Observe that the quantity $\diff{}n f(x \mid \vec u)$ is actually symmetric in $\vec u$, i.e. stable under permutations of the coordinates of $\vec u$.
\begin{definition}
A function $f: \boule \coneone \rightarrow \conetwo$ is called a \emph{stable function from $\coneone$ to $\conetwo$} if it is pre-stable, sequentially Scott-continuous, and moreover there exists $\lambda \in \RR_+$ such that $f(\boule \coneone) \subseteq \lambda \cdot \boule \conetwo$.
  \end{definition}

\begin{definition}
$\cstab$ is the category whose objects are sequentially complete cones, and morphisms from $\coneone$ to $\conetwo$ are the stable functions $f$ from $\coneone$ to $\conetwo$ such that $f(\boule \coneone) \subseteq \boule \conetwo$.
  \end{definition}
It was shown in~\cite{pse} that it is possible to endow $\cstab$ with a cartesian closed structure.
The product cone is defined as $\prod_{i \in I} \coneone_i = \{(x_i)_{i \in I} \mid \forall i \in I, x_i \in \coneone_i \}  $, and $\norm{\prod_{i \in I}\coneone_i}{x} = \sup_{i \in I} \norm {\coneone_i}{x_i}$.
The function cone $\coneone \Rightarrow \conetwo$ is the set of all stable functions, with $\norm {\coneone \Rightarrow \conetwo} f = \sup_{x \in \boule \coneone} \norm \conetwo {f(x)}$. It was also shown in~\cite{pse} that these cones are indeed sequentially complete, and that the lub in $\coneone \Rightarrow \conetwo$ is computed pointwise. We will use also the cone of pre-stable functions from $\coneone$ to $\conetwo$, which is also sequentially complete.
\hide{
\begin{figure}[!h]
\begin{center}
  \fbox{
  \begin{minipage}{0.45 \textwidth}
    \footnotesize
    \begin{center}
      ${\prod_{i \in I} \coneone_i = \{(x_i)_{i \in I} \mid \forall i \in I, x_i \in \coneone_i \} \quad \norm{\prod_{i \in I}\coneone_i}{x} = \sup_{i \in I} \norm {\coneone_i}{x_i}} $
       ${\coneone \Rightarrow \conetwo = \{ f: \boule \coneone \rightarrow \conetwo \text{ stable functions } \} \quad \norm {\coneone \Rightarrow \conetwo} f = \sup_{x \in \boule \coneone} \norm \conetwo {f(x)}}$
    \end{center}
  \end{minipage}
}
\end{center}
\caption{Cartesian Structure of $\cstab$.}
  \end{figure}
}
\subsection{A generalization of Bernstein's theorem for pre-stable functions}\label{subsect:bernstein}

We are now going to show an analogue of Bernstein's Theorem for pre-stable functions on directed-complete cones.
The idea is to first define an analogue of derivatives for pre-stable functions, and to show that pre-stable functions can be written as the infinite sum generated by an analogue of Taylor expansion on $\bouleopen \coneone$. This result is actually an application of McMillan's work~\cite{mcmillan} in the setting of abstract domains. Here, we give the main steps of the construction directly on cones, and highlight some properties of the Taylor series which are true for cones, but not in the general framework McMillan considered.

\subsubsection{Derivatives of a pre-stable function}
We are now going, following McMillan~\cite{mcmillan}, to construct derivatives for pre-stable functions on directed complete cones. This construction is based on the use of a notion of \emph{partition}: a \emph{partition} of $x \in \boule \coneone$ is a multiset $\pi = [u_1, \ldots,u_n] \in \mfin{\coneone}$ such that $x = \sum_{1 \leq i \leq n} y_i $. We write $\partit {\pi} x$ when the multiset $\pi$ is a partition of $x$. We will denote by $+$ the usual union on multiset: $[y_1, \ldots, y_n] + [z_1, \ldots,z_m] = [y_1, \ldots, y_n,z_1, \ldots,z_m]$. We call $\parts x$ the set of partitions of $x$.

\begin{definition}[Refinement Preorder]
  If $\pi_1$, $\pi_2$ are in $\parts x$, we says that $\pi_1 \leq \pi_2$ if $\pi_1 = [u_1, \ldots, u_n]$, and $\pi_2 = \alpha_1 + \ldots + \alpha_n$ with each of the $\alpha_i$ a partition of $u_i$.
\end{definition}

Observe that when $\pi_1$ and $\pi_2$ are partition of $x$, $\pi_2 \leq \pi_1$ means that $\pi_1$ is a more \emph{finely grained} decomposition of $x$. 
If $\vec u$ is an $n$-tuple in $\boule \coneone$, we extend the refinement order to $ \parts {\vec u} = \parts{u_1} \times \ldots \times \parts{u_n}$.
\begin{lemma}\label{lemma:refinmentdirected}
  Let be $\coneone$ a lattice cone.
  Then for every $x \in \coneone$, $\parts x$ is a directed set.
\end{lemma}
\shortv{The proof of Lemma~\ref{lemma:refinmentdirected} may be found in the long version. }
\longv{\begin{proof}
  We are going to use the following notion: we say that two non-zero elements $x$ and $y$ of $\coneone$ are \emph{orthogonal}, and we note $x \perp y$, if $x \wedge y =0$.
  Let be $\pi_1, \pi_2 \in \parts x$. We first show that it cannot exist $z \in \pi_1$ which is orthogonal to all the element of $\pi_2$. Indeed, suppose that it is the case: we take $y_1, \ldots,y_n$ such that $\pi_2 = [y_1, \ldots,y_n]$. Then by hypothesis, $z \leq x = \sum_{1 \leq i \leq n} y_i$. We can now use the decomposition property from Lemma~\ref{lemma:lattice}. It means that $z = \sum_{1\leq i \leq n} z_i$, with $z_i \leq y_i$. But since for all $i$, $z \perp y_i$, it folds that $z_i=0$ for all $i$, and so $z=0$, and we have a contradiction.

  Now, we are going to present a procedure to construct a partition $\pi$ of $x$ with $\pi \leq \pi_1$, and $\pi \leq \pi_2$. We can suppose that all elements of $\pi_1$ and $\pi_2$ are non-zero. We start form $\pi=[]$, $\theta_1 = \pi_1$, $\theta_2= \pi_2$, and $w = x$, $v=0$. Through the procedure, we guarantee:
  \begin{itemize}
  \item $\theta_1, \theta_2 \in \parts w$, $\pi \in \parts v$, and $w + v = x$;
  \item  all the elements of $\Theta_1$ and $\Theta_2$ are non-zero;
  \item  $\pi + \Theta_1 \leq \pi_1$, and $\pi + \Theta_2 \leq \pi_2$ (for the refinment order).
\end{itemize}
    Then at each step of the procedure, if $\theta_1$ is non empty, we do the following: let $\theta_1 = [a_1, \ldots, a_n]$, and $\theta_2 = [b_1, \ldots ,b_m]$. Then we know that there is a $j$, such that $a_1$ and $b_j$ are not orthogonal. We modify the variables as follows:
  \begin{align*}
    \pi &= \pi + [a_1 \wedge b_j] \\
    v &= v + a_1 \wedge b_j \\
    \theta_1 &= \begin{cases}[a_1 - a_1 \wedge b_j, a_2, \ldots a_n] \text{ if } a_1 \wedge b_j \neq a_1\\
      [a_2, \ldots, a_n] \text{ otherwise.}
      \end{cases}\\
    \theta_2 &= \begin{cases}
      [b_1, \ldots, b_{j-1}, b_j - a_1 \wedge b_j, b_{j+1}, \ldots, b_m] \text{ if }a_1 \wedge b_j \neq b_j\\
      [b_1, \ldots, b_{j-1},  b_{j+1}, \ldots, b_m] \text{ otherwise.}
      \end{cases}\\
  x &= x - a_1 \wedge b_j
  \end{align*}
  At every step of the procedure presented above, the quantity:
  $$\card{(i,j) \mid \text{ not }(a_i\perp b_j)  } $$
  decreases. Indeed:
  \begin{itemize}
  \item or we remove either $a_1$ of $\Theta_1$, or $b_j$ of $\Theta_2$, and then the statement above holds.
  \item or we replace $a_1$ by $(a_1 - a_1 \wedge b_j)$, and $b_j$ by $(b_j - a_1 \wedge b_j)$. Then we see that $(a_1 - a_1 \wedge b_j) \perp (a_1 - a_1 \wedge b_j)$. Moreover, the pairs that were orthoganal before are still orthogonal: indeed for every $z$ with $z \perp a_1$ it holds that $z \perp a_1 - a_1 \wedge b_j$, and the same for $b_j$.  
  \end{itemize}
  As a consequence, the procedure will terminates. It means that we reach a state where $\Theta_1$ is empty, and all the invariants presented above hold. Then we see that $\pi \in \parts x$, and $\pi \leq \pi_1, \pi_2$.

  We are going to illustrate the procedure above on a very basic example. We consider the cone consisting of the positive quadrant of $\RR^2$, endowed by the order defined as: $x \leq y$ if $x_1 \leq y_1$, and $x_2 \leq y_2$. We take two partitions of a vector $x \in \RR^2$: $\pi_1 = [b_1, b_2]$, and $\pi_2 = [a_1, a_2]$, where $a_1, a_2,b_1,b_2$ are taken as pictured in Figure~\ref{fig:firststepprod}. We are going to apply our procedure in order to obtain a refinment of both $\pi_1$ and $\pi_2$. At the beginning, we have $\Theta_1 = \pi_1$, $\Theta_2 = \pi_2$, $w=x$, $v=0$.
  \begin{itemize}
  \item The first step is represented in Figure~\ref{fig:firststepprod}. Observe that the procedure is actually non-deterministic: we may choose any $(a,b)$ with $a \in \pi_1$, $b \in \pi_2$, and $a$ and $b$ not orthonal. Here, we choose to start from $(b_1, a_1)$. We take $v = a_1 \wedge b_1$ (and we represent it by a red vector in Figure~\ref{fig:firststepprod}): it is going to be the first element of our new partition $\pi$. Accordingly, we take $\pi = [v]$. We know update the partition $\Theta_1$ and $\Theta_2$ into partitions of $w = x-v$:  $\Theta_2$ becomes $[a'_1, a_2]$, and  $\Theta_1$ becomes $[b'1,b_2]$ where $a'_1 = a_1 - b_1 \wedge a_1$ and $b'_1 = b_1 - a_1 \wedge b_1$ are represented also in red in Figure~\ref{fig:firststepprod}.
  \item The second step is represented in Figure~\ref{fig:sndstepprod}. Observe that now $a'_1$ and $b'_1$ are orthogonal, so we have to choose another pair. Here, we choose $(b'_1, a_2 )$. As before, we add to $\pi$ the glb of $b'_1$ and $a_2$: we obtain $\pi = [a_1 \wedge b_1, b'_1 \wedge a_2]$. Observe that now (as can be seen on Figure~\ref{fig:sndstepprod}, $b'_1 \leq a_2$, and so $b'_1 \wedge a_2 = b'_1$. So when we update the partition $\Theta_1$ and $\Theta_2$, we take:
    $\Theta_2 = [b_2]$, and $\Theta_1=[a'1,a'_2]$ where $a'_2 = a_2 - b'1$ is represented in purple in Figure~\ref{fig:sndstepprod}.
    \item By doing again two steps of the procedure, we see that the final partition $\pi$ is $[a_1 \wedge b_1, b'_1, a'_1, a'_2]$. We cen see by looking at Figure~\ref{fig:sndstepprod} that it is indeed a refinment of both  $\pi_1$ and $\pi_2$.. 
    \end{itemize}

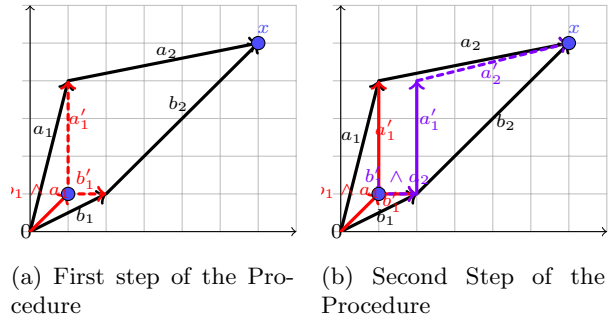
\begin{figure}[!h]
    \centering
    \begin{subfigure}[b]{0.25\textwidth}
        \definecolor{ffqqqq}{rgb}{1.,0.,0.}
\definecolor{ududff}{rgb}{0.30196078431372547,0.30196078431372547,1.}
\definecolor{cqcqcq}{rgb}{0.7529411764705882,0.7529411764705882,0.7529411764705882}
\begin{tikzpicture}[line cap=round,line join=round,x=0.5cm,y=0.5cm]
\draw [color=cqcqcq,, xstep=0.5cm,ystep=0.5cm] (-0.2,0.) grid (7.,6.);
\draw[->,color=black] (0.,0.) -- (7.,0.);
\draw[->,color=black] (0.,0.) -- (0.,6.);

\draw[color=black] (-0.5,0) node[right] {\footnotesize $0$};
\clip(-0.5,0.) rectangle (7.,7.);
\draw [->,line width=1.2pt] (0.,0.) -- (1.,4.);
\draw [->,line width=1.2pt] (1.,4.) -- (6.,5.);
\draw [->,line width=1.2pt] (0.,0.) -- (2.,1.);
\draw [->,line width=1.2pt] (2.,1.) -- (6.,5.);
\draw [->,line width=1.2pt,color=ffqqqq] (0.,0.) -- (1.,1.);
\draw [->,line width=1.2pt,color=ffqqqq,dash pattern=on 2pt off 2pt] (1.,1.) -- (1.,4.);
\draw [->,line width=1.2pt,color=ffqqqq,dash pattern=on 2pt off 2pt] (1.,1.) -- (2.,1.);
\begin{scriptsize}
\draw [fill=ududff] (6.,5.) circle (2.5pt);
\draw[color=ududff] (6.14,5.37) node {$x$};
\draw[color=black] (0.36,2.65) node {$a_1$};
\draw[color=ffqqqq] (1.3,3) node {$a'_1$};
\draw[color=black] (3.6,4.75) node {$a_2$};
\draw[color=black] (1.46,0.43) node {$b_1$};
\draw[color=ffqqqq] (1.46,1.4) node {$b'_1$};
\draw[color=black] (3.9,3.39) node {$b_2$};
\draw [fill=ududff] (1.,1.) circle (2.5pt);
\draw[color=ffqqqq] (0.2,1.1) node {$b_1 \wedge a_1$};
\end{scriptsize}
\end{tikzpicture}
        \caption{First step of the Procedure}
        \label{fig:firststepprod}
    \end{subfigure}
    ~ 
    \begin{subfigure}[b]{0.25\textwidth}
      \input
      \definecolor{xfqqff}{rgb}{0.4980392156862745,0.,1.}
\definecolor{ffqqqq}{rgb}{1.,0.,0.}
\definecolor{ududff}{rgb}{0.30196078431372547,0.30196078431372547,1.}
\definecolor{cqcqcq}{rgb}{0.7529411764705882,0.7529411764705882,0.7529411764705882}
\begin{tikzpicture}[line cap=round,line join=round,x=0.5cm,y=0.5cm]
\draw [color=cqcqcq,, xstep=0.5cm,ystep=0.5cm] (-0.2,0.) grid (7.,6.);
\draw[->,color=black] (0.,0.) -- (7.,0.);
\draw[->,color=black] (0.,0.) -- (0.,6.);
\draw[color=black] (-0.5,0) node[right] {\footnotesize $0$};
\clip(-0.5,0.) rectangle (8.,8.);
\draw [->,line width=1.2pt] (0.,0.) -- (1.,4.);
\draw [->,line width=1.2pt] (1.,4.) -- (6.,5.);
\draw [->,line width=1.2pt] (0.,0.) -- (2.,1.);
\draw [->,line width=1.2pt] (2.,1.) -- (6.,5.);
\draw [->,line width=1.2pt,color=ffqqqq] (0.,0.) -- (1.,1.);
\draw [->,line width=1.2pt,color=ffqqqq] (1.,1.) -- (1.,4.);
\draw [->,line width=1.2pt,color=ffqqqq] (1.,1.) -- (2.,1.);
\draw [->,line width=1.2pt,dash pattern=on 2pt off 2pt,color=xfqqff] (1.,1.) -- (2.,1.);
\draw [->,line width=1.2pt,color=xfqqff] (2.,1.) -- (2.,4.);
\draw [->,line width=1.2pt,dash pattern=on 2pt off 2pt,color=xfqqff] (2.,4.) -- (6.,5.);
\begin{scriptsize}
\draw [fill=ududff] (6.,5.) circle (2.5pt);
\draw[color=ududff] (6.14,5.37) node {$x$};
\draw[color=black] (0.29,2.5) node {$a_1$};
\draw[color=black] (3.43,4.96) node {$a_2$};
\draw[color=black] (1.19,0.38) node {$b_1$};
\draw[color=black] (4.35,2.9) node {$b_2$};
\draw [fill=ududff] (1.,1.) circle (2.5pt);
\draw[color=ffqqqq] (0.2,1.1) node {$b_1 \wedge a_1$};
\draw[color=ffqqqq] (1.19,2.7) node {$a'_1$};
\draw[color=ffqqqq] (1.35,0.8) node {$b'_1$};
\draw[color=xfqqff] (1.5,1.46) node {$b'_1 \wedge a_2$};
\draw[color=xfqqff] (2.35,2.98) node {$a'_1$};
\draw[color=xfqqff] (3.97,4.2) node {$a'_2$};
\end{scriptsize}
\end{tikzpicture}
        \caption{Second Step of the Procedure}
        \label{fig:sndstepprod}
    \end{subfigure}
    \caption{Illustration of the Proof of Lemma\ref{lemma:refinmentdirected}}\label{fig:refinment}
\end{figure}

\end{proof}}
Observe that, as a consequence, the refinement preorder turnsalso $\parts {\vec u}$ into a directed set.
\begin{definition}[from~\cite{mcmillan}]\label{def:dersn1}
  Let $\coneone$ be a lattice cone, $\conetwo$ a cone, and let $f: \boule \coneone \rightarrow \conetwo$ be a pre-stable function. Then for every $x \in \boule \coneone$, and $\vec u = (u_1, \ldots ,u_n)\in \boule \coneone_x^n $,
  we define $\Phi_{x,\vec u}^{f,n} : \parts {\vec{u}} \rightarrow \conetwo$ as:
  $$\Phi_{x,\vec u}^{f,n} (\pi_1, \ldots \pi_n) = \sum_{y_1 \in \pi_1} \ldots \sum_{y_n \in \pi_n} \diff{} n f (x \mid y_1, \ldots, y_n) .$$
\end{definition}

It holds (see~\cite{mcmillan} for more details) that $\Phi_{x, \vec u}^{f,n}$ is a non-increasing function whenever $f$ is pre-stable\longv{ (it is shown in Lemma 3.2 of~\cite{mcmillan} by looking at the definition of higher-order differences)}. Since $\parts {\vec{u}}$ is a directed set, $\Phi_{x,\vec u}^{f,n}$ has a greatest lower bound whenever $\conetwo$ is a directed-complete lattice cone.

\begin{definition}[from~\cite{mcmillan}]\label{def:dersn}
  Let be $\coneone$ a lattice cone, $\conetwo$ a directed-complete lattice cone, and $f: \boule \coneone \rightarrow \conetwo$ a pre-stable function. 
  Let be $\vec u \in \boule \coneone_x^n $. Then the \emph{derivative of $f$ in $x$ at rank $n$ towards the direction $\vec u$} is the function $\ders n f(x \mid \cdot): \boule{\coneone_x^n} \rightarrow \conetwo$ defined as
  $$\ders n f (x \mid \vec u) = \inf_{\vec \pi \in \parts{\vec{u}} }  \Phi_{x,\vec u}^f (\vec \pi).$$
\end{definition}

We are now going to illustrate Definition~\ref{def:dersn} on a basic case where we take $f:\Rp{} \rightarrow \Rp{}$, in order to highlight the link with differentiation in real analysis.
\begin{example}\label{ex:derdiff}
  We take $\coneone$ and $\conetwo$ as the positive real half-line, and $x \in [0,1[$. Let be $h$ such that $x + h \leq 1$. Then:
      $$\ders 1 f (x \mid h) =\inf_{\pi \text{ with } \partit {\pi}{h} } \sum_{y \in \pi} f(x+y) - f(x) $$
      We know already, since $f$ is pre-stable hence absolutely monotone as function on reals, that $f$ is convex, and moreover differentiable (see~\cite{widder}). From there, by considering a particular family of partitions, we can show that $\ders 1 f (x \mid h) = h \cdot f'(x)$\shortv{ (the proof can be found in the long version)}.
      \longv{
        \begin{proof}
        First, let $\pi$ be any partition of $y$. Since $f$ is differentiable and convex, it holds that:
        $$\forall z , \, f(x+z) - f(x) \geq f'(x)\cdot z. $$
        As a consequence, we see that for any partition $\pi$ of $h$, it holds that $ \sum_{y \in \pi} f(x+y) - f(x) \geq f'(x) \cdot h$, and it implies that $\ders 1 f (x \mid h) \geq f'(z)\cdot h$. To show the reverse inequality, it is enough to consider the particular family of partition $\pi_n = [\frac h n, \ldots, \frac h n]$ of $h$: we see that
\begin{align*}
  \sum_{y \in \pi_n} f(x+y) - f(x) & = n \cdot f(x + \frac{h}{n} )- f(x)\\ & = h \cdot \frac{ f(x + \frac{h}{n}) - f(x)}{\frac h n} \rightarrow_{n \rightarrow \infty} h \cdot f'(x).
  \end{align*}
        \end{proof}
        }
\end{example}
\begin{lemma}\label{lemma:aux1}
  Let be $\coneone$ a lattice cone, $\conetwo$ a directed complete cone, $f$ a pre-stable function from $\coneone$ to $\conetwo$. Let be $x \in \bouleopen \coneone$.
  Then $\ders n f(x \mid \cdot)$ is a symmetric function $\boule{({\coneone_x^n})} \rightarrow \conetwo$ such that moreover:
      \begin{itemize}
        \item $0 \leq \ders n f (x \mid \vec u) \leq \diff{}n f (x \mid \vec u)$.
        \item Both $\vec u \mapsto \ders n f (x \mid \vec u )$ and $\vec u \mapsto \diff {} n f(x \mid \vec u) - \ders n f (x \mid \vec u)$ are pre-stable functions from $\coneone_x^n$ to $\conetwo$. 
  \end{itemize}
\end{lemma}
\begin{proof}
The proof is given in Lemma 3.31 in~\cite{mcmillan}. It comes almost directly from Definition~\ref{def:dersn}.
      \end{proof}

We have seen in Example~\ref{ex:derdiff} that our so-called derivatives of pre-stable functions play the same role as the differential of a differentiable function, which are actually \emph{linear} operators $df_x^n: \mathbb{R}^n \rightarrow \mathbb{R}$. While the abstract domains considered in~\cite{mcmillan} do not have to be $\RR_+$ semi-modules, so have no notion of linearity, we are able to show in the complete cone case\shortv{ (see the long version for the proof)} that the $\ders n f$ are linear in the sense of Lemma~\ref{lemma:additivity} below.
\begin{lemma}\label{lemma:additivity}
  Let $\coneone$, $\conetwo$ be two  directed complete lattice cones,  $x \in \bouleopen \coneone$.
  \begin{itemize}
  \item Let $f:\boule \coneone \rightarrow \conetwo$ be a pre-stable function. Then $\ders n f(x \mid \cdot) : \boule {(\coneone_x^n)} \rightarrow \conetwo$ is $n$-linear, in the sense that, for each of its arguments, it commutes with the sum and multiplication by a scalar.
  \item For any $\vec u \in  \boule {(\coneone_x^n)}$, the function $f \in \cstab(\coneone, \conetwo) \mapsto \ders n f(x \mid \vec u) \in \conetwo$ is linear and directed Scott-continuous.
  \end{itemize}
     \end{lemma}
\longv{\begin{proof}
    We are going to use the following auxiliary lemma:
    \begin{lemma}[from~\cite{mcmillan}]\label{lemma:scdersf}
      Let $\coneone$ and $\conetwo$ be two directed cones, and $f:\coneone \rightarrow \conetwo$ linear and non-decreasing, such that moreover for all subset $F$ of $\coneone$ directed for the reverse order, $f(\inf F) = \inf f(F)$. Then $f$ is directed Scott-continuous.
    \end{lemma}
    \begin{proof}
      Let be $E$ a directed subset of $\coneone$.
      We define $F = \{ \sup E - x \mid x \in E \}$. Since $E$ is directed, $F$ is directed for the reverse order, and as a consequence:$\inf f(F) = f(\inf F)$. But we see that $\inf F = 0$. Therefore, since $f$ is linear, $f(\inf F) = 0$. As a consequence (and again by linearity of $f$): $f(\sup E) - \sup f(E) = \inf f(F) = 0$.
    \end{proof}
    We are now going to show Lemma~\ref{lemma:additivity}.
    \begin{itemize}
\item We first show that $\ders n f(x \mid \cdot) : \boule {(\coneone_x^n)} \rightarrow \conetwo$ is $n$-linear.
  The additivity is given by Lemma 3.72 of ~\cite{mcmillan}. The commutation with scalar multiplication  is not proved on this form in ~\cite{mcmillan} because they have a more general notion of a system of partition. We first show that the result holds when $\lambda$ is a rational number. To do that, we use the fact that $\pi = [\frac x n , \ldots, \frac x n]$ is always a partition of $x$. Then, let $\lambda \in \Rp{}$ and $\vec u =(u_1, \ldots,u_n)$ such that both $\vec u$ and $\vec v = (u_1, \ldots, \lambda u_i, \ldots ,u_n)$ are in $\boule {\coneone_x^n}$. Let be $\overline r = (r_m)_{m \in \NN}, \overline q = (q_m)_{m \in \NN}$ two sequences of rational number such that $\overline r$ tends to $\lambda$ by below, and $\overline q$ tends to $\lambda$ by above. We see that:
  $$\ders n f(x \mid \vec v) = 2\cdot \ders n f (x \mid u_1, \ldots, \frac \lambda 2 \cdot u_i, \ldots u_n) .$$
  We take $N$ such that for every $m \geq N$, $q_m \leq 2 \cdot \lambda$: since $\ders n f(x\mid \cdot)$ is non-decreasing, we see that:
  \begin{align*}
  \ders n f &(x \mid u_1, \ldots, \frac {r_m} 2 \cdot u_i, \ldots ,u_n) \\&\leq 
  \ders n f (x \mid u_1, \ldots, \frac \lambda 2 \cdot u_i, \ldots, u_n) \\ &\leq
  \ders n f (x \mid u_1, \ldots, \frac {q_m} 2 \cdot u_i, \ldots ,u_n)
  \end{align*}
  Applying now the linearity for rational numbers, we see that for every $m \geq N$:
  \begin{align*}
  r_m \cdot \ders n f &(x \mid u_1, \ldots, \frac {1} 2 \cdot u_i, \ldots ,u_n) \\&\leq 
  \ders n f (x \mid u_1, \ldots, \frac \lambda 2 \cdot u_i, \ldots, u_n) \\ &\leq
  q_m \cdot \ders n f (x \mid u_1, \ldots, \frac {1} 2 \cdot u_i, \ldots ,u_n)
  \end{align*}
  As a consequence:
   \begin{align*}
     \sup_{m \geq N} r_m \cdot \ders n f &(x \mid u_1, \ldots, \frac {1} 2 \cdot u_i, \ldots ,u_n) \\&
     \leq \ders n f (x \mid u_1, \ldots, \frac \lambda 2 \cdot u_i, \ldots, u_n) \\ &\leq \inf_{m \in N}
  q_m \ders n f (x \mid u_1, \ldots, \frac {1} 2 \cdot u_i, \ldots ,u_n)
   \end{align*}
   and by Scott-continuity of $\cdot$, it tells us that $\ders n f (x \mid u_1, \ldots, \frac \lambda 2 \cdot u_i, \ldots, u_n) = \lambda \ders n f (x \mid u_1, \ldots, \frac {1} 2 \cdot u_i, \ldots ,u_n) $. We can now conclude: recall that $
   \ders n f(x \mid \vec v) = 2\cdot \ders n f (x \mid u_1, \ldots, \frac \lambda 2 \cdot u_i, \ldots u_n) .$  Therefore:
   \begin{align*}
     \ders n f(x \mid \vec v) &=2 \cdot \lambda \cdot \ders n f (x \mid u_1, \ldots, \frac {1} 2 \cdot u_i, \ldots ,u_n)\\
     & = \lambda \cdot \ders n f(x \mid \vec u)  \quad \text{ since }\frac 1 2 \in \mathbb{Q}.
     \end{align*}
\item We show now that $f \in \cstab(\coneone, \conetwo) \mapsto \ders n f(x \mid \vec u) \in \conetwo$ is linear and Scott-continuous. It is immediate that it is linear, since every one of the $f \mapsto \diff {}n f{(x\mid u)}$ is. We are now going to use~\ref{lemma:scdersf} to show the Scott-continuity: it tells us that we have only to check that for every $E \subseteq \coneone \Rightarrow_m \conetwo$ directed for the reverse order, $\ders n f (\inf E \mid \vec u) = \inf \ders n f (E \mid \vec u)$. Observe that:
  \begin{align*}
    \ders n {(\inf E)} (x \mid \vec u) &= \inf_{\vec \pi \in \parts{\vec{u}} } \sum_{y_1 \in \pi_1} \ldots \sum_{y_n \in \pi_n} \diff {}n {(\inf E)} (x \mid y_1, \ldots, y_n) \\
    &=  \inf_{\vec \pi \in \parts{\vec{u}} } \sum_{y_1 \in \pi_1} \ldots \sum_{y_n \in \pi_n} \inf_{f \in E} \{\diff{} n {f} (x \mid y_1, \ldots, y_n)\}\\
    &=  \inf_{\vec \pi \in \parts{\vec{u}} } \inf_{f \in E}\{\sum_{y_1 \in \pi_1} \ldots \sum_{y_n \in \pi_n} \diff{}n {f} (x \mid y_1, \ldots, y_n)\}\\
    & = \inf_{f \in E} \ders n f (x \mid \vec u) \text{ since the infs can be exchanged}.
    \end{align*}
  \end{itemize}
\end{proof}}
\shortv{The proof of Lemma~\ref{lemma:additivity} can be found in the long version.}
The linearity of the derivatives means that for every $x \in \bouleopen \coneone$, we can extend $\ders n f(x \mid \cdot)$ to a function $\coneone_n^x \rightarrow \conetwo$. We will use implicitly this extension in the following, especially in Definition~\ref{def:TSn}. 
\subsubsection{Taylor Series for pre-stable functions}
We have seen above that the $\ders n f$ are a notion of differential for pre-stable functions. Following further this idea, and the work of McMillan~\cite{mcmillan}, we define an analogue to the Taylor expansion.  In all this section $\coneone$ and $\conetwo$ are going to be directed complete lattice cones, and $f: \boule \coneone \rightarrow \conetwo$ a pre-stable function. 
\begin{definition}~\label{def:TSn}
       Let be $x \in \bouleopen \coneone$.
        We call \emph{ Taylor partial sum of $f$ in $x$ at the rank $N$} the function
           $Tf^N(x \mid \cdot): \boule {\coneone_x^1} \rightarrow \conetwo$ defined as:   $$Tf^N(x \mid  y) = f(x) +  \sum_{k=1}^N \frac 1 {k!} \ders k f (x \mid y, \ldots, y) .$$
\end{definition}
The next step consists in establishing that the $T^N f$ are actually a non-increasing bounded sequence in the cone of pre-stable functions from $\coneone$ to $\conetwo$, which will allow us to define the \emph{Taylor series of $f$}, as the supremum of the $T^Nf$ (see the long version for more details on the proof). 

\longv{
To that end, we are first going to establish an alternative characterization of the Taylor series, which is the one used in~\cite{mcmillan}, in the framework of abstract domains.  It consists in substituting each of the $\ders n f(x \mid y, \ldots, y)$ with its expression given by Lemma~\ref{lemma:dersn} below. The validity of Lemma~\ref{def:dersn}, and thus the equivalence of the two definitions, depends on the fact we work with directed-complete cones.

\begin{lemma}[Alternative Caracterisation of Derivatives]\label{lemma:dersn}
  Let $x \in \bouleopen \coneone, y \in \boule {\coneone_x^1}$, and $k \in \NN$. Then it holds that
$\ders k f (x \mid y, \ldots, y) $ is equal to: $$\sup_{\pi = [u_1, \ldots, u_n]  \in \parts y}\sum_{\sigma:[\![ 1, k ]\!] \hookrightarrow [\![ 1, n ]\!] } \ders k f (x \mid u_{\sigma(1)}, \ldots u_{\sigma(n)}) $$ 
\end{lemma}
\shortv{The proof of Lemma~\ref{lemma:dersn} can be found in the long version.}
\longv{\begin{proof}
    We first introduce the following notation: if $\pi = (u_1, \ldots, u_n)$ is a partition of $x$, and  $\sigma:[\![ 1, k ]\!] \hookrightarrow [\![ 1, n ]\!]$ an injective function, we denote $\sigma(\pi) = (u_{\sigma(1), \ldots, u_\sigma(k)})$.
  We denote by $A = \sup_{\pi  \in \parts y}\sum_{\sigma:[\![ 1, k ]\!] \hookrightarrow [\![ 1, n ]\!] } \ders k f (x \mid {\sigma(\pi)})$.
  We show separately the two inequalities.
  \begin{itemize}
  \item We first show that $A \geq \ders k f (x \mid y, \ldots, y)$.  For every $n \in \NN$, it holds that $\pi = (\frac 1 n \cdot y, \ldots, \frac 1 n \cdot y)$ is a partition of $y$. Therefore for every $n \in \NN$:
\begin{align*}
  A &\geq \sum_{\sigma:[\![ 1, k ]\!] \hookrightarrow [\![ 1, n ]\!] } \ders k f (x \mid \frac y n, \ldots, \frac y n)\\
  &  = \frac{n!}{(n-k)!} \ders k f (x \mid \frac y n, \ldots, \frac y n)\\
  &=  \frac{n!}{(n-k)! \cdot n^k}  \ders k f (x \mid y, \ldots, y)
\end{align*}
The sequence $\frac{n!}{(n-k)! \cdot n^k}$ tends to $1$ when $n$ tends to infinity (see in the long version). By Scott-continuity, it means that $A \geq \ders k f (x \mid y, \ldots, y)$.
\item Let us show now that $A \leq \ders k f (x \mid y, \ldots, y)$. Let be $\pi = (u_1, \ldots, u_n) \in \parts y$. Then:
\begin{align*}
  & \sum_{\sigma:[\![ 1, k ]\!] \hookrightarrow [\![ 1, n ]\!] } \ders k f (x \mid \sigma(\pi)) \\
  & \leq \sum_{i_1 \in \{1, n\}} \ldots \sum_{i_k \in \{1, \ldots, n\}} \ders k f (x \mid u_{i_1}, \ldots u_{i_k}) \\
  & = \ders k f(x \mid y, \ldots, y) \text{ by }n\text{-linearity of }\ders k f(x \mid \cdot)
  \end{align*}  
Since $A = \sup_{\pi  \in \parts y}\sum_{\sigma:[\![ 1, k ]\!] \hookrightarrow [\![ 1, n ]\!] } \ders k f (x \mid {\sigma(\pi)}) $, we see that $A \leq \ders k f(x \mid y, \ldots, y) $, which ends the proof. 
\end{itemize}
\end{proof}
 With this characterization,~\cite{mcmillan} shows that the sequence of functions $(x \in \boule \coneone_y^1 \mapsto Tf^n(x \mid y))$ is bounded  by $(x \in \boule \coneone_y^1 \rightarrow f(x + y))$ in the cone of pre-stable functions from $\coneone_x^1$ to $\conetwo$. } }
\begin{lemma}\label{lemma:tfN}
Let be $y$ is in $\bouleopen \coneone$, and $x$ in $\boule\coneone_y^1$. Then $\forall N \in \NN$, $Tf^N(x \mid y) \leq f(x + y)$, and the function  $(x \in \boule\coneone_y^1 \mapsto f(x + y) - Tf^N(x \mid y))$ is pre-stable.
\end{lemma}
\longv{
\begin{proof}
  Let be $x \in \bouleopen \coneone$ and  $y \in\boule\coneone_x$. We are able to express $f(x+y)$ by using $f(x)$ and finite differences on any partition of $y$: indeed, for every partition $\pi$ of $y$, it holds that:
   \begin{equation}\label{eq:idf}
 f(x + y) = f(x) + \sum_{1 \leq k\leq n} \frac 1 {k!}\sum_{\sigma:[\![ 1, k ]\!] \hookrightarrow [\![ 1, n ]\!] } \diff{} k f (x, \sigma(\pi))
   \end{equation}
   \shortv{
(The proof of~\eqref{eq:idf} is done in~\cite{mcmillan} in a purely combinatorial way, by induction on the cardinal of the partition $\pi$. It can be found in the long version for the case $n=2$.)}
   \longv{
\begin{proof}
  It is an algebraic calculation, done in~\cite{mcmillan}. We give here the proof for $n=2$. Let $\pi = [y_1, y_2]$ a partition of $y$. Then we see that:
\begin{align*}
  f(x) &+ \sum_{1 \leq k\leq n} \frac 1 {k!}\sum_{\sigma:[\![ 1, k ]\!] \hookrightarrow [\![ 1, n ]\!] } \diff{} k f (x, \sigma(\pi)) \\
  &= f(x) + \diff{} 1 f (x,y_1) + \diff{}1 f(x,y_2)  \\
  &  \qquad +\frac 1 2 \cdot (\diff{}2 f(x,y_1,y_2) + \diff{}2 f(x,y_2,y_1) )\\
  & = f(x) + (f(x+y_1) - f(x)) + (f(x+y_2)-f(x)) \\
  &  \qquad + (f(x+y_1+y_2) - f(x+y_1) - f(x+y_2)+f(x)) \\
  & = f(x+y_1+y_2) = f(x+y).
  \end{align*}
  \end{proof}
     }
   Moreover we are also able to express the derivatives of $f$ at $x$ towards the direction $y$ also using the partitions of $y$ (it is the sense of Lemma~\ref{lemma:dersn}). Accordingly:
   \begin{align*}
     &Tf^N(x \mid y) = f(x) +  \sum_{k=1}^N \frac 1 {k!} \ders k f (x \mid y, \ldots, y) \\
     &= f(x) +  \sum_{k=1}^N \frac 1 {k!} \sup_{\pi \in \parts y}\sum_{\sigma:[\![ 1, k ]\!] \hookrightarrow [\![ 1, n ]\!] } \ders k f (x \mid \sigma(\pi))
   \end{align*}
   Using Lemma~\ref{lemma:aux1}, we see that it implies:
   $$Tf^N(x \mid y) \leq f(x) +  \sum_{k=1}^N \frac 1 {k!} \sup_{\pi \in \parts y}\sum_{\sigma:[\![ 1, k ]\!] \hookrightarrow [\![ 1, \#(\pi) ]\!] } \diff{} k f (x \mid \sigma(\pi))$$
   We can now use the Scott continuity of $+$ and $\cdot$, and we obtain: 
     $$Tf^N(x \mid y) \leq \sup_{\pi \in \parts y} f(x) +  \sum_{k=1}^N \frac 1 {k!}\sum_{\sigma:[\![ 1, k ]\!] \hookrightarrow [\![ 1, \#(\pi) ]\!] } \diff{} k f (x \mid \sigma(\pi)) $$
     We can now conclude using~\eqref{eq:idf}:
     $$Tf^N(x\mid y) \leq \sup_{\pi \in \parts y} f(x+y) \leq f(x+y)$$
     
The proof of the pre-stability of the function can be found in~\cite{mcmillan}. It is based on the fact that each one of the above inequality can be seen as an inequality in the cone of pre-stable functions.
   \end{proof}}
Since we have shown that the partial sum of the Taylor series of $f$ was a bounded non-decreasing sequence in the complete cone of pre-stable functions from $\coneone_x^1$ to $\conetwo$, we can now define the \emph{Taylor series of $f$} as its supremum.
\begin{definition}~\label{def:TS}
  We define $Tf(x \mid \cdot): \boule {\coneone_x^1} \rightarrow \conetwo$ \emph{the Taylor series of $f$ in $x$}, and $Rf(x \mid \cdot): \boule {\coneone_x^1} \rightarrow \conetwo$ \emph{the Remainder of $f$ in $x$} as:
  \begin{align*}
    Tf(x \mid y) &= \sup_{N \in \NN}Tf^N(x \mid y)\\
    Rf(x \mid y) &= f(x + y) - Tf(x \mid y).
    \end{align*}
\end{definition}

 \subsubsection{Extended Bernstein's theorem}
 Our goal from here is to show that for any $x \in \bouleopen \coneone$, $Rf(0 \mid x) = 0$. We recall here the main steps of the proof of~\cite{mcmillan}. It is based on two technical lemmas, that analyze more precisely the behavior of the remainder of $f$. \longv{The first one is actually a summary of several technical results shown separately in~\cite{mcmillan}.}\shortv{ The proofs are done in~\cite{mcmillan}, and a sketch can be found in the long version.}

 \begin{lemma}\label{lemma:auxb1}
    Let be $x \in \bouleopen \coneone$.
    Then it holds that both:
    \begin{align*}
      Rf_x : y \in \boule{\coneone_x^1} &\mapsto Rf(x \mid y)  \in \conetwo \\
      \text{and }\quad Rf^y: x \in \boule{\coneone_y^1} &\mapsto Rf(x\mid y)  \in \conetwo 
      \end{align*}
    are pre-stable functions. Moreover $Rf_x(0) = 0$, and
for every $x \in \boule{\coneone_y^1}$, it holds that $T(Rf^y)(0 \mid x) = 0$.
 \end{lemma}
 \longv{\begin{proof}
     We give here only sketches of the proofs. The detailed proof can be found in~\cite{mcmillan}.
   \begin{itemize}
   \item For $Rf^y$, it is a consequence of the fact that both $f^y:(x \in \boule \coneone_y^1 \mapsto f(x+y) $ and $Tf^y:x \in \boule \coneone_y^1 \rightarrow Tf(x \mid y)$ are pre-stable functions, with $Tf^y \leq f^y$ in the cone of pre-stable functions, and $Rf^y = f^y - Tf^y$.
  \item The pre-stability of $Rf_x$ is stated in Theorem 4.1. of~\cite{mcmillan}. It is based on a previous technical lemma shown in~\cite{mcmillan}, which says it is sufficient for a function to be pre-stable, to have all its differences \emph{in 0} to be non-negative. Then the idea is to fix $x$, and to consider for every $N \in \NN$, the function $g_N:\, y \in \boule \coneone_x^1 \mapsto f(x+ y) - Tf^N(x \mid y)$. It is then possible to show that for any $n \in \NN$, and $\vec u \in \boule\coneone_y^n$, $\diff{} n {g_N} (0 \mid \vec u) = \diff{} n f(x \mid \vec u) -  \diff{} n (Tf^N_x)(0 \mid \vec u)$, with $Tf^N_x: y \mapsto Tf^N(x \mid y)$. By a computation on the $\diff{} n (Tf^N_x)(0 \mid \vec u)$, we see that the $\diff {} n {g_N}(0 \mid \vec u)$ are non-negative. Then, we conclude using the fact that $Rf_x(y) = \inf_{N \in \NN} Tf^N_x(y)$. 
\item The fact that $Rf_x(0) = 0$ is a direct consequence of the $n$-linearity of the map $\vec u \mapsto \ders n f(x \mid \vec u)$ for $n \geq 1$.
  \item The fact that $T(Rf^y)(0 \mid x) =0$ is shown in~\cite{mcmillan} in Lemma 5.26. It is based on the fact that the Scott-continuity of $f \mapsto \ders n f (x \mid \vec u)$ allows us to show that if we take $g(y) = \ders n f(x\mid \vec u)$, then $\ders k g(x \mid \vec v) = \ders {n+k} f(y_0 \mid \vec u, \vec v)$, and from there to compute the Taylor series of $Rf^y$.
     \end{itemize}
 \end{proof}}
 The second technical lemma gives us a way to decompose $Rf(x \mid y)$ into smaller pieces. It is stated in Theorem 5.3 in~\cite{mcmillan}.
\longv{
 \begin{lemma}\label{lemma:auxb2}
   Let be $x, y$ such that $x + y \in \boule \coneone$. Then $Rf(0 \mid x+y) \leq Rf(y \mid x) + Rf( x \mid y)$, and furthermore $Rf(0 \mid x+y) \geq Rf(x \mid y)$, and $Rf(0\mid x+y) \geq Rf(y \mid x)$, and moreover all the inequality are in the cone of pre-stable functions.
 \end{lemma}}
\longv{ 
   \begin{proof}
     We give here a brief sketch of the proof of the first statement. More details can be found in~\cite{mcmillan}. We introduce the function $Rf^{+ x}: y \in \boule\coneone_x^1 \mapsto Rf(0 \mid x+y)$.
     The proof is based on the fact that it is possible to establish (see~\cite{mcmillan}):
     \begin{align}
       T(Rf^{+ x})(0 \mid y) &= T(Rf^x)(0 \mid y) \label{eq:proofmckey1}\\
       \text{ and }R(Rf^{+ x})(0 \mid y) &=R (Rf_x)(0 \mid y) \label{eq:proofmckey2}
     \end{align}
     As a consequence, we can write:
     \begin{align*}
       &Rf(x \mid y) + Rf(y \mid x) = Rf_x(y) + Rf^x(y)\\
       & = T(Rf_x)(0 \mid y) + R(Rf_x)(0 \mid y) + T(Rf^x)(0 \mid y) + R(R_f^x)(0 \mid y) \\
       & = T(Rf_x)(0 \mid y) + R(Rf^{+ x})(0 \mid y) + T(Rf^{+ x})(0 \mid y) + R(R_f^x)(0 \mid y)\\
       & \quad \text{by~\eqref{eq:proofmckey1} and~\eqref{eq:proofmckey2}}\\
       & = Rf^{+x}(y) + T(Rf_x)(0 \mid y) +  R(R_f^x)(0 \mid y) \\
       & \geq Rf^{+x}(y) = Rf(0 \mid x+y),
     \end{align*}
     and we see that we have also shown that the difference is pre-stable. The other two statement are shown in a similar way.
     \end{proof}
 
 We use Lemma~\ref{lemma:auxb2} to show the a more involved upper bound on $Rf(0 \mid x)$.}
 \begin{lemma}\label{lemma:auxbern1}
   Let be $x \in \boule \coneone$, and $\pi = [x_1, \ldots, x_n]$ a partition of $x$, such that for every $x_i \in \pi$, $x + x_i \in \boule \coneone$. Then $Rf(0 \mid x) \leq \sum_{1 \leq i \leq n} \inf_{\pi_i \mid \partit{\pi_i} x_i} \sum_{z \in \pi_i} Rf(x \mid z).$
 \end{lemma}
\longv{ \begin{proof}
   For every $x \in \boule\coneone$, we denote $g_x: y \in \boule \coneone_x^1 \mapsto f(x+y)$. From the definitions of the $\ders n {}{}$, we see that it holds that $Rg_x(0 \mid y) = Rf(x \mid y)$.
   
         Let be $\pi_1, \ldots, \pi_n$ such that $\pi_i$ is a
         partition of $x_i$ over $J$. Then $\pi_1 + \ldots +\pi_n$ is a partition of $x$.
Lemma~\ref{lemma:auxb2} applied several times , combined with the fact that  $Rg_x(0 \mid y) = Rf(x \mid y)$, tells us that:
         $$Rf(0 \mid x_0) \leq \sum_{z \in \pi_1 + \ldots +\pi_n} Rf (z' \mid z), $$
         where $z' = \sum_{u \in \pi_1 + \ldots +\pi_n \mid u \neq z } u $.
         Moreover, we know that $Rf^z$ is pre-stable (by lemma~\ref{lemma:auxb1}). Since, for every $z \in \pi_1 + \ldots +\pi_n$, $z' \leq x$ (it is immediate, since $\pi_1 + \ldots +\pi_n$ is a partition of $x$), it folds that $ Rf(z' \mid z)  =Rf^z(z')  \leq Rf^z(x) = Rf (x \mid z) $.
                  As a direct consequence, we see that
         $Rf(0 \mid x_0) \leq \sum_i \sum_{z \in \pi_i} Rf (x \mid z), $
         which leads to the result.
       \end{proof}}
 We are now ready to show the main result of this section.
   \begin{proposition}[Extended Bernstein's Theorem]\label{prop:ebt}
   Let be $\coneone$, $\conetwo$ directed-complete lattice cones, and $f: \boule \coneone \rightarrow \conetwo$ a pre-stable function. 
   Then for every $x \in \bouleopen{{\coneone}}$, it holds that $f(x) = Tf(0 \mid x)$.
   \end{proposition}
   
   \begin{proof}
     Let be $x \in \boule \coneone$.
     First, we consider the partition $\pi = [\frac x N, \ldots, \frac x N]$ of $x$, with $N$ taken such as $x + \frac x N \in \boule\coneone$. We know that such an $N$ exists since $x$ is in the open unit ball $\bouleopen \coneone$.
     We use Lemma~\ref{lemma:auxbern1} on $Rf(0\mid x)$, and the partition $\pi$, and it tells us that:
\begin{equation}\label{eq:auxeb2}
  Rf(0 \mid x) \leq \sum_{1 \leq j \leq N}\inf_{\pi=(u_1, \ldots, u_n) \in \parts {\frac x N}} \sum_{1 \leq i \leq n} Rf (x  \mid u_i) .
\end{equation}
Observe that  the above expression is valid, since for every $u_i$ in a partition $\pi$ of $\frac x N$, $x+u_i \in \boule \coneone$.
We know, by Lemma~\ref{lemma:auxb1} that $Rf(x \mid 0) = 0$. 
Therefore, we can rewrite~\eqref{eq:auxeb2} as:
     \begin{equation}\label{eq:auxeb3}
     Rf(0 \mid x) \leq \sum_{1 \leq j \leq N} \inf_{\pi=(u_1, \ldots, u_n) \in \parts {\frac x N}}\sum_{1 \leq i \leq n} Rf (x  \mid u_i) - Rf (x \mid 0) .
     \end{equation}
     Moreover, we are able to express the right part of~\eqref{eq:auxeb3} by the finite differences of the pre-stable function $Rf_{x}$; indeed for each $i$,
     \begin{equation}\label{eq:auxeb4}
       \diff {}1 {Rf_{x} }(0,u_i) = Rf (x  \mid u_i) - Rf (x \mid 0) .
       \end{equation}
     By the definition of derivatives (see Definition~\ref{def:dersn}), we see that
     \begin{equation}\label{eq:auxeb5}
       \ders 1 {Rf_{x}}(0 \mid \frac x N) = \inf_{\pi \in \parts {\frac x N}} \sum_{v \in \pi} \diff{} 1 {Rf_{x}}(0,v)
       \end{equation}
     
     We see now that combining~\eqref{eq:auxeb3},~\eqref{eq:auxeb4} and~\eqref{eq:auxeb5} leads us to $ Rf(0 \mid x) \leq \sum_{1 \leq j \leq N}\ders 1 {Rf_{x}}(0 \mid \frac x N).$
     Moreover, we know that for every $y \in \boule\coneone_x^1$, $\ders 1 {Rf_{x}}(0 \mid y) \leq T(Rf_{x} )(0 \mid y).$
     Hence by using again Lemma~\ref{lemma:auxb1}, which says that $T(Rf_{x})(0 \mid {\frac x N}) = 0$, it holds that $Rf(0 \mid x) = 0$.
     
  \end{proof}

\section{$\cstab$ is a conservative extension of $\pcoh_!$}\label{sect:pcs}

Probabilistic coherence spaces (PCS) were introduced by Ehrhard and Danos in~\cite{pcsaamohpc} as a model of higher-order probabilistic computation. It was successful in giving a fully abstract model both of $\pcf_\oplus$, and of a discrete probabilistic extension of Levy's Call-by-Push-Value. In this section, we present briefly basic definitions from~\cite{pcsaamohpc} and highlight an embedding from PCSs into cones.

\subsection{Probabilistic Coherence Spaces}
The definition of the PCS model of Linear Logic follows the tradition initiated by Girard with Coherence Spaces in~\cite{girard1986system}, and followed for instance by Ehrhard in~\cite{Ehrhard93} when defining hypercoherence spaces. A coherent space interpreting a type can be seen as a symmetric graph, and the interpretation of a program of this type is a \emph{clique} of this graph. Interestingly, such a graph $A$ can be alternatively characterized by giving its set of vertices (that we will call \emph{web}), and a family of subsets of this web, meant to be the family of the cliques of $A$. Then we know that an arbitrary family of subsets of a given web arises indeed as a family of cliques for some graph when some \emph{duality criterion} is verified.

 PCSs are designed to express probabilistic behavior of programs. As a consequence, a clique is not a subset of the web anymore, but a \emph{quantitative} way to associate a non-negative real coefficient to every element in the web. 
\begin{definition}[Pre-Probabilistic Coherent Spaces]\label{def:prepcs}
A Pre-PCS is a pair $\pcsone = (\web \pcsone, \prog \pcsone)$, where $\web \pcsone$ is a countable set called \emph{web of $\pcsone$},  $\prog \pcsone $ is a subset of $\subseteq \Rp{\web \pcsone}$ whose elements are called \emph{cliques of $\pcsone$}. 
\end{definition}
We need here to introduce some notations to deal with infinite dimensional $\RR$-vector spaces. Given a countable web $A$, and $a$ an element of $A$, we denote $e_a$ the vector in $\Rp A$ which is $1$ in $a$, and $0$ elsewhere. We are also going to introduce a scalar product on vectors in $\Rp A$: if $u, v \in \Rp A$, we will denote $\scal u v = \sum_{a \in \web \pcsone} u_a v_a \in \mathbb{R} \cup \{\infty\} $. Moreover, if $A$ and $B$ are countable sets, $x \in \Rp{A \times B}$, and $u \in \Rp{A}$, we denote by $x \cdot u$ the vector in $(\Rp{}\cup \{\infty\})^{B}$ given by $(x \cdot u)_b = \sum_{a \in A} x_{a,b} u_a$ for every $b \in B$.

We are going to give examples of pre-PCS modeling discrete data-types. First, we define a pre-PCS $\unit$ to correspond to unit type. Since unit-type programs have only one possible outcome (that they can reach or not), $\unit$ has only one vertex: $\web \unit = \{\star\}$. We want the denotation of a unit-type program to express its probability of termination: we take the set of cliques $\prog \unit$ as the interval $[0,1]$. 

Let us now look at what happens when we consider programs of type $N$: a program can now have a countable numbers of possible outcomes, so the web will consist of $\NN$, and cliques will be sub-distributions on these vertices.
\begin{example}[Pre-PCS of Natural Numbers]\label{ex:Nat}
We define the Pre-PCS $\NN^{\pcoh}$ by taking $\web \NN^\pcoh = \NN$, and $\prog {\NN^\pcoh} = \{u \in \Rp {\NN} \mid_{n \in \NN} \sum u_n \leq 1\}$. It corresponds to the denotational semantics of the base type $N$ of $\pcf_\oplus$ in $\pcoh_!$.
  \end{example}

We now need to give a \emph{quantitative} bi-duality criterion, to specify which one of the $\prog \pcsone \subseteq \Rp{\web \pcsone}$ are indeed \emph{valid} families of cliques. To do that, we first define a \emph{duality operator}: if $\pcsone = (\web \pcsone , \prog \pcsone)$ is a pre-PCS, we define the pre-PCS  $\dual \pcsone =  (\web \pcsone, \{u \in \Rp {\web \pcsone}, \, \forall v \in \prog \pcsone, \scal u v \leq 1 \})$. We are now ready to give conditions on pre-PCSs to actually be PCS:
\begin{definition}[Probabilistic Coherent Spaces]\label{def:pcs}
  A pre-PCS $\pcsone$ is a \emph{PCS} if $\dual {\dual \pcsone} = \pcsone$ and moreover the following two conditions hold:
  \begin{itemize}
  \item $ \forall a \in \web \pcsone$, there exists $\lambda >0$ such that $\lambda e_a \in \prog \pcsone$.
    \item $\forall a \in \web \pcsone$, there exists $M \geq 0$, such that  for every $u \in \prog \pcsone$, $u_a \leq M$.
    \end{itemize}
\end{definition}
We may see easily that both $\unit$ and $\NN^{\pcoh}$ are indeed PCSs.

 As highlighted in Example 4.4 from~\cite{pse}, we can associate in a generic way a cone to any PCS. The idea is that we consider the extension of the space of cliques by all uniform scaling by positive reals. We formalize this idea in Definition~\ref{def:pcstocone} below.
\begin{definition}\label{def:pcstocone}
  Let be $\pcsone$ a PCS. We define a cone $\pcstocone \pcsone$ as the $\Rp{}$ semi-module $\{\alpha \cdot x \text{ s.t. } \alpha \geq 0, x\in \prog \pcsone  \}$ where the $+$ is the usual addition on vectors. We endow it with $\norm {\pcstocone \pcsone}{\cdot}$ defined by:
  $$
    \norm {\pcstocone \pcsone}{x} = \sup_{y \in \prog{\dual \pcsone}} {\scal x y}  = \inf\{\frac 1 r \mid \, r \cdot x \in \prog \pcsone \}.$$
\end{definition}
It is easily seen that it is indeed a cone (the proof uses the so-called technical conditions from Definition~\ref{def:pcs}). Moreover, we can see that $\boule{\pcstocone \pcsone} $ consists exactly of the set $\prog \pcsone$ of cliques of $\pcsone$.
Looking at the cone order $\order{\pcstocone \pcsone}$, as defined in Definition~\ref{def:coneorder}, we see that it coincides on $\prog \pcsone$ with the pointwise order in $\Rp{\web \pcsone}$. It is relevant since we know already from~\cite{pcsaamohpc} that $\prog \pcsone$ is a bounded-complete and $\omega$-continuous cpo with respect to this pointwise order.

\begin{lemma}\label{lemma:pcsdccones}
For every PCS $\pcsone$, it holds that $\pcstocone \pcsone$ is a directed-complete lattice cone.
\end{lemma}
\begin{proof}
To show that $\pcstocone \pcsone$ is directed complete, we use the fact that $\prog \pcsone$ is a complete partial order. To show that it is a lattice, we see that $x \vee y $ can be defined as: $(x \vee y)_a = \max{x_a,y_a} \, \forall a \in \web \pcsone$. 
\end{proof}

\subsection{The Category $\pcoh$.}
Intuitively a morphism in $\pcoh(\pcsone,\pcstwo)$ is a linear map from $\Rp{\web \pcsone}$ to $\Rp{\web \pcstwo}$ \emph{preserving} the cliques.
\begin{definition}[Morphisms of PCSs]\label{def:morphismspcs}
Let be $\pcsone$, $\pcstwo$ two PCSs. A \emph{morphism of PCSs between $\pcsone$ and $\pcstwo$} is a matrix $x \in \Rp{\web \pcsone \times \web \pcstwo}$ such that for every $u \in \prog \pcsone$, it holds that $x \cdot u \in \prog \pcstwo$.
\end{definition}

We now illustrate Definition~\ref{def:morphismspcs} by looking at the morphisms from $\Bool$ to itself: they are the $x \in \Rp{\{\ttrue, \ffalse\} \times \{\ttrue, \ffalse\}}$ such that ${x_{\ttrue, \ttrue} + x_{\ttrue, \ffalse } \leq 1}$, and similarly ${x_{\ffalse, \ttrue} + x_{\ffalse, \ffalse } \leq 1}$. We see that they are exactly those matrices specifying the transitions for a probabilistic Markov chain with two states $\ttrue$ and $\ffalse$.

We call $\pcoh$ in the following the category of PCS and morphisms of PCS.
In~\cite{pcsaamohpc}, it is endowed with the structure of a model of linear logic. We are only going to recall here partly the exponential structure, since our main focus will be on the Kleisli category associated to $\pcoh$. 

In~\cite{pcsaamohpc}, the construction of the exponential was done by defining a functor $!$, as well as dereliction and digging making $\pcoh$ a Seely category, and consequently a model of linear logic. Here, we are only going to recall explicitly the effect of $!$ on PCSs. We denote by $\mfin {\web \pcsone}$ the set of finite multisets over the web of $\pcsone$, and we take it as the web of the PCS $\bang \pcsone$. If $\mu \in \mfin A$, we call \emph{support of $\mu$}, and we denote $\supp \mu$, the set of elements $a$ is $A$ such that $a$ appears in $\mu$. Moreover, we will use the following notation: for every $x \in \Rp{\web \pcsone}$, and $\mu \in \mfin{\web \pcsone}$, we denote $x^\mu = \prod_{a \in \supp \mu} x_a^{\mu(a)} \in \Rp{}$.
\begin{definition}
  Let be $\pcsone$ a PCS. We define the \emph{promotion} of $x \in \prog \pcsone$, as the element $x^! \in \Rp{\mfin{\web \pcsone}}$ given by $x^!_\mu = x^\mu.$
We define $!\pcsone = (\mfin {\web \pcsone}, \{x^! \mid x \in \pcsone\}^{\bot \bot})$.
  
\end{definition}

\subsection{The Kleisli Category of Probabilistic Coherence Spaces}\label{subsect:pcoh!}
The idea, as usual, is that morphisms in the Kleisli category can use several times their argument, while morphisms in the original category are \emph{linear}.
The Kleisli category for $\pcoh$, denoted $\pcoh_!$, has also PCSs for objects, while $\pcoh_!(\pcsone,\pcstwo) = \pcoh(\bang \pcsone, \pcstwo)$. 
We give here a direct characterization of $\pcoh_!$ morphisms.
\begin{lemma}[from~\cite{pcsaamohpc}]\label{lemma:morphpcoh}
Let be $f \in \Rp{\mfin{\web {\pcsone}} \times {\web\pcstwo}}$. Then $f$ is a morphism in $\pcoh_!(\pcsone, \pcstwo)$, if and only if for every $x \in \prog \pcsone$, $f \cdot {x^!} \in \prog \pcstwo$.
  \end{lemma}
What Lemma~\ref{lemma:morphpcoh} tells us is that any $f \in \pcoh_!(\pcsone, \pcstwo)$ is entirely characterized by the map $\mapm f : x \in  \prog \pcsone \rightarrow f \cdot x^! \in \prog \pcstwo$. We denote by $\mathcal E^{\pcsone, \pcstwo}$ the set of all maps $\prog \pcsone \rightarrow \prog \pcstwo$ that are equal to a $\mapm f$ with $f \in \pcoh_!(\pcsone, \pcstwo)$. It has been shown in~\cite{pcsaamohpc} that $\mapm{(\cdot)}$ is actually a bijection from $\pcoh_!(\pcsone, \pcstwo)$ to $\mathcal E^{\pcsone, \pcstwo}$.

Observe that we can see the maps in $\mathcal E^{\pcsone, \pcstwo}$ as \emph{entire series}, in the sense that they can be written as the supremum of a sequence of polynomials. Indeed, for any morphism $f$, and $x \in \prog \pcsone$, we can write:
$$\mapm f (x) = \sup_{N \in \NN} {\sum_{b \in \web \pcstwo} (\sum_{\mu \text{ with }\card \mu \leq N} f_{\mu,b } \cdot x^\mu ) \cdot e_b} $$

As the Kleisli category of the comonad $!$ in a Seely category, $\pcoh_!$ is a cartesian closed category. We give here explicitely the construction of the product and arrow constructs: if $\pcsone$ and $\pcstwo$ are PCSs, $\pcsone \Rightarrow \pcstwo$ is defined by  
$\web{\pcsone \Rightarrow \pcstwo} = \mfin {\web\pcsone} \times \web \pcstwo$ and $\prog{(\pcsone \Rightarrow \pcstwo )} = \pcoh_!(\pcsone, \pcstwo)$. If $(\pcsone_i)_{i \in I}$ is a family of PCSs, $\prod_{i \in I} \pcsone_i$ is defined by
    $ \web{\prod_{i \in I}{\pcsone_i}} = \cup_{i \in I}\{i\} \times \web {\pcsone_i}$ and $ \prog \pcsone = \{x \in \Rp{\web{\prod_{i \in I}{\pcsone_i}}} \mid \forall i \in I, \, \pi_i(x) \in \prog{\pcsone_i} \}$, where $\pi_i(x)_{a} = x_{(i,a)}$.
  
\subsection{A fully faithful functor $\functor: \pcoh_! \rightarrow \cstab$.}\label{sect:proofconsext}

Recall that Definition~\ref{def:pcstocone} gave a way to see a PCS as a cone. Moreover, as stated in Proposition~\ref{prop:fmsf} below, a morphism in $\pcoh_!$ can also be seen as a stable function, in the sense that $\mathcal E^{\pcsone, \pcstwo} \subseteq \cstab(\pcstocone \pcsone, \pcstocone \pcstwo)$.
\begin{proposition}\label{prop:fmsf}
Let be $f \in \pcoh_!(\pcsone, \pcstwo)$. Then $\mapm f$ is a stable function from ${\pcstocone \pcsone}$ to $\pcstocone \pcstwo$.
  \end{proposition}
\begin{proof}
  We know from~\cite{pcsaamohpc} that $\mapm f: \prog \pcsone \rightarrow \prog \pcstwo$ is sequentially Scott-continuous with respect to the orders $\order{\pcstocone \pcsone}$, $\order{\pcstocone \pcstwo}$.
  Moreover $\mapm f$ is pre-stable: it comes from the fact that $\mapm f$ can be written as a power series with all its coefficients non-negative.
  Finally, we have to show that $\mapm f(\boule \pcstocone \pcsone) \subseteq \boule \pcstocone \pcstwo)$. Since $\boule \pcstocone \pcsone = \prog \pcsone$, $\boule \pcstocone \pcsone = \prog \pcsone$, and moreover $f$ is a morphism in $\pcoh_!(\pcsone, \pcstwo)$, we see that the result holds. 
  \end{proof}
Thus we can define a functor $\functor: \pcoh_! \rightarrow \cstab$, by taking $\functor \pcsone = \pcstocone \pcsone$, and $\functor f = \mapm f$. Our goal now is to show that $\functor$ is full and faithful, which will make $\pcoh_!$ a full subcategory of $\cstab$.
As mentioned before, it was shown in~\cite{pcsaamohpc} that $\,\mapm{\cdot}\,$ is a bijection from $\pcoh_!(\pcsone, \pcstwo)$ to $\mathcal{E}^{\pcsone, \pcstwo}$. It tells us directly that $\functor$ is indeed \emph{faithful}.
In the remainder of this section, we are going to show that $\functor$ is actually also \emph{full}, hence makes $\cstab$ a conservative extension of $\pcoh_!$.

In the following, we fix $\pcsone$ and $\pcstwo$ two PCSs, and $g \in \cstab(\functor \pcsone, \functor \pcstwo)$. Our goal is to show that there exists $f \in \pcoh_!(\pcsone, \pcstwo)$ such that $\mapm f = g$.
First, recall that we have shown in Lemma~\ref{lemma:pcsdccones} that for every PCS $\pcsthree$, the cone $\functor \pcsthree$ is a directed complete lattice cone. It means that all results in Section~\ref{subsect:bernstein} can be used here: in particular, $g$ has higher-order derivatives $D^n g$, which makes Definition~\ref{def:fullfunctor} below valid.

\begin{definition}\label{def:fullfunctor}
We define $f\in \Rp{\mfin{\web {\pcsone}} \times \web \pcstwo}$ by taking:
$${f}_{[a_1, \ldots, a_k], b} = \frac {\alpha_{[a_1, \ldots, a_k]}}{k!}\left(\ders k g(0 \mid e_{a_1}, \ldots e_{a_k} ) \right)_b \in \RR^+.$$
where 
$\alpha_\mu = \# {\{(c_1, \ldots,c_k) \in \web \pcsone^k \text{ with }\mu = [c_1, \ldots, c_k] \}}.$
\end{definition}
We have to show now that $f \in \pcoh_!(\pcsone, \pcstwo)$, and that $\mapm {f}$ coincides with $g$ on $\prog \pcsone$. The key observation here is that we have actually built $f$ in such a way that it is going to coincide with $Tg(0 \mid \cdot)$\textemdash the Taylor series of $g$ defined in Definition~\ref{def:TS}. We first show it for the elements of $\prog \pcsone$ with \emph{finite support}, by using finite additivity of the $\ders k g (0 \mid \cdot)$. 
\begin{lemma}\label{lemma:densetaylor}
  Let be $x \in \prog{\pcsone}$, such that $\text{Supp}(x) = \{a \in \prog \pcsone \mid x_a > 0\}$ is finite. Then it holds that
 $ f \cdot x^!$ is finite (i.e for every $b \in \web \pcstwo$, $(f \cdot x^!)_b < \infty$), and moreover  $f \cdot x^!  = Tg(0 \mid x)$.
\end{lemma}
\begin{proof}
  Let $A= \{a_1, \ldots, a_m\} \subseteq \web \pcsone$ be the set $\text{Supp}(x)$. For any $b \in \web \pcstwo$, we can deduce from the definition of $f$ that:
  \begin{align*}
    (f & \cdot x^!)_b = \sum_{k=0}^{\infty} \sum_{\mu = [c_1, \ldots, c_k] \in \mfink k {A}}\frac {\alpha_\mu}{k !}\cdot \ders {k} g(0 \mid e_{c_1}, \ldots e_{c_k} )_b \cdot x^\mu,
  \end{align*}
  where $\mfink k A$ stands for the set of multisets over $A$ of cardinality $k$.
Looking at the definition of $\alpha_\mu$, we see that this implies:
  \begin{equation}\label{eq:pcsbern1}
     (f \cdot x^!)_b  = \sum_{k=0}^{\infty} \sum_{(c_1, \ldots, c_k) \in  A^k}\frac {1}{k !} \ders {k} g(0 \mid e_{c_1}, \ldots e_{c_k} )_b \cdot \prod_{i=1}^k x_{c_i} 
  \end{equation}
  By Lemma~\ref{lemma:additivity}, we know that $\ders k g(0 \mid \cdot)$ is $k$-linear. As a consequence, and since $x = \sum_{i=1}^m x_{c_i}\cdot e_{c_i}$ and that moreover $A$ is \emph{finite}, we see that~\eqref{eq:pcsbern1} implies the result:
  $$(f \cdot x^!)_b  = \sum_{k=0}^{\infty} \frac {1}{k !} \ders {k} g(0 \mid x, \ldots x )_b  = (Tg(0 \mid x))_b.$$
\end{proof}

We are now going to apply the generalized Bernstein's theorem, as stated in Proposition~\ref{prop:ebt}, to the stable function $g$ from $\functor \pcsone$ to $\functor \pcstwo$. It tells us that:
\begin{equation}\label{eq:bernkey}
  \forall x \in \bouleopen{\pcstocone \pcsone}, \quad g( x) = Tg(0 \mid x).
  \end{equation}


Combining ~\eqref{eq:bernkey} with Lemma~\ref{lemma:densetaylor}, we obtain that:
\begin{equation}\label{eq:eqdense}
  \forall x \in \bouleopen{\pcstocone \pcsone} \text{ with } \text{Supp}(x) \text{ is finite},\, f \cdot x^! = g(x).
  \end{equation}

We can now use~\eqref{eq:eqdense} to show that $\mapm f$ and $g$ coincide on $\prog \pcsone$: the key point is that the subset of elements in $\prog \pcsone$ of norm smaller than $1$ and finite support is dense, and that moreover $g$ is Scott-continuous.
\begin{lemma}\label{lemma:fullnessaux2}
   $\forall x \in \prog \pcsone, \, f \cdot x^! = g(x)$, and moreover $f \in \pcoh_!(\pcsone, \pcstwo)$.
\end{lemma}
\begin{proof}
  
  Let be $x \in \prog \pcsone$. We define a sequence $(y_n)_n \in \NN$ of elements in $\prog \pcsone$, by taking:
  $$(y_n)_a = \begin{cases}
    (1 - \frac 1 {2^n})\cdot x_a \text{ if }\lambda(a) \leq n\\
    0 \text{ otherwise,}
  \end{cases}$$
  where we have fixed $\lambda$ an arbitrary enumeration of the elements of $\web \pcsone$\textemdash $\lambda$ exists since it is a countable set. Observe the the sequence $(y_n)_{n \in \NN}$ is non-decreasing, with $x = \sup_{n \in \NN} y_n$. 
  Moreover, for every $n$, $y_n$ has finite support and $\norm{\pcstocone \pcsone}{y_n} < 1$ .
  It means that for every $y_n$, we can use~\eqref{eq:eqdense}: we see that $g(y_n) = f \cdot y_n^!$. Since $g$ is a morphism in $\cstab$, $g$ is sequentially Scott-continuous, hence:
\begin{equation}\label{eq:auxf1}
  g(x) = \sup_{n \in \NN} g(y_n).
\end{equation}
  
Moreover, we know from~\cite{pcsaamohpc} that both $x \mapsto x^!$ and $x \mapsto u \cdot x$ are Scott continuous. It means that:
\begin{equation}\label{eq:auxf2}
  f\cdot x^! = \sup_{n \in \NN} f \cdot {y_n^!}.
  \end{equation}
  Combining~\eqref{eq:auxf1} and~\eqref{eq:auxf2}, we obtain $f \cdot x^! = g(x)$.
  Since  $g(\boule{\pcstocone \pcsone}) \subseteq \boule{\pcstocone \pcstwo}$, it implies also that $\mapm f(\prog \pcsone)\subseteq \prog \pcstwo$. Thus by Lemma~\ref{lemma:morphpcoh}  $f \in \pcoh_!(\pcsone, \pcstwo)$.
\end{proof}
Since we have indeed been able to show in Lemma~\ref{lemma:fullnessaux2} that for any fixed stable function $g$ in $\cstab(\functor \pcsone, \functor \pcstwo)$, there exists an $f \in \pcoh_!(\pcsone, \pcstwo)$ such that $\functor f = g$, we have indeed shown that $\functor$ is full.
\subsection{$\functor$ preserves the cartesian structure.}
We want now to give a stronger guarantee on the functor $\functor$: we want to show that it is a \emph{cartesian closed functor}, meaning that it embeds the cartesian closed category $\pcoh_!$ into the cartesian closed category $\cstab$ in such a way that:
\begin{itemize}
\item $\functor$ preserves the product: for every family $(\pcsone_i)_{i \in I}$ of PCSs, $\functor{(\prod_{i \in I}^{\pcoh_!} \pcsone_i)}$ is isomorphic to $\prod_{i \in I}^{\cstabm} \functor{\pcsone_i}$;
  \item $\functor$ preserves function spaces: for every $\pcsone, \pcstwo$ PCSs, $\functor{(\pcsone \Rightarrow \pcstwo)}$ is isomorphic to $\functor \pcsone \Rightarrow \functor \pcstwo$. 
  \end{itemize}

\begin{lemma}\label{lemma:fpresprod}
$\functor$ preserves cartesian products.
\end{lemma}
\begin{proof}
  We fix a family $\fone = (\pcsone_i)_{i \in I}$ of PCSs.
  In order to construct an isomorphism, we have a canonical candidate, given by:
\begin{equation}\label{eq:morphismcartproduct}
  \Psi^{\fone} = \langle \functor (\pi_i) \mid i \in I \rangle \in \cstab( \functor{(\prod_{i \in I}^{\pcoh_!} \pcsone_i)} ,  \prod_{i \in I}^{\cstabm} \functor{\pcsone_i}) .
  \end{equation}
Let us now see that $\Psi^{\fone}$ is an isomorphism, \shortv{i.e. that it has an inverse}.\longv{ Looking at the definition of cartesian product in $\pcoh_!$ defined in Section~\ref{subsect:pcoh!}, and the one of cartesian product in $\cstab$, defined in Section~\ref{subsect:cstab}, we see that for every $x \in \boule \functor{(\prod_{i \in I}^{\pcoh_!} \pcsone_i)} $:
$$\Psi^\fone(x) = (y_i)_{i \in I} \quad \text{where} \quad \forall i \in I, \forall a \in \web {\pcsone_i}, \,  (y_i)_a = x_{(i,a)} .$$
We want now to show that $\Psi^{\fone}$ has an inverse.} The only candidate is $\Theta^{\fone}: y \in \boule (\prod_{i \in I}^{\cstab} \functor{\pcsone_i}) \mapsto \Theta(y) \in {(\functor{(\prod_{i \in I}^{\pcoh_!} \pcsone_i)} )}$, defined by:
$\forall i \in I, a \in \web{\pcsone_i}, \Theta(y)_{i,a} = (y_i)_a .$
 We see immediately that $\Theta^{\fone}$ is linear, hence pre-stable, and that moreover it is Scott-continuous. Besides, it is also preserves the unit ball, since $\forall y \in \boule \conetwo, \, \norm {\functor{(\prod_{i \in I}^{\pcoh_!} \pcsone_i)}}{\Theta^{\fone}(y)} = \norm {\prod_{i \in I}^{\cstab} \functor{\pcsone_i}}y$\shortv{ (the proof can be found in the long version)}.\longv{
We show now that $\Theta^\fone$ preserves the unit ball:
 \begin{align*}
  &\norm{\coneone}{\Theta^{\fone}(y)} =  \inf\{\frac 1 r \mid r \cdot \Theta^{\fone}(y) \in \prog \prod_{i \in I}^{\pcoh_!} \pcsone_i \} \\
  &\quad =\inf\{\frac 1 r \mid  \forall i \in I, r \cdot y_i \in \prog  \pcsone_i \} 
  =\sup_{i \in I} \norm{\functor \pcsone_i}{y_i} 
  =\norm{\prod_{i \in I}^{\cstab} \functor \pcsone_i} y. 
  \end{align*} }
Thus $\Theta^\fone$ is a morphism in $\cstab$.
\end{proof}
\begin{lemma}\label{lemma:fpresarr}
 $\functor$ preserves function spaces.
\end{lemma}
\begin{proof}
Let $\pcsone, \pcstwo$ two PCSs. As previously, there is a canonical candidate for the isomorphism: we define $\Upsilon^{\pcsone, \pcstwo}$ as the currying in $\cstab$ of the morphism:
$$ \functor{(\pcsone \Rightarrow \pcstwo)} \times \functor \pcsone\stackrel{\Theta^{\pcsone \Rightarrow \pcstwo, \pcsone}}{\xrightarrow{\hspace*{1.2cm}}} \functor(\pcsone \Rightarrow \pcstwo \times \pcsone) \stackrel{\functor(\text{eval}_{\pcsone, \pcstwo})}{\xrightarrow{\hspace*{1.3cm}}} \functor \pcstwo,$$
where $\Theta^{\pcsone \Rightarrow \pcstwo, \pcsone}$ is as defined in the proof of Lemma~\ref{lemma:fpresprod} above.

Unfolding the definition, we see that actually:
$\Upsilon^{\pcsone, \pcstwo}: f \in \boule{\functor(\pcsone \Rightarrow \pcstwo)} \mapsto \mapm f \in (\functor \pcsone \Rightarrow \functor \pcstwo)$. 
Since we have shown that $\functor$ is full and faithful, we can consider $\Xi^{\pcsone, \pcstwo}$ the inverse function of $\Upsilon^{\pcsone, \pcstwo}$.
Recall from the proof of the fullness of $\functor$ in Section~\ref{sect:proofconsext} that for every $\mu = [a_1, \ldots, a_k] \in \mfin {\web \pcsone}$, and $b \in \web \pcstwo$:
$${\Xi^{\pcsone, \pcstwo}(f)}_{\mu, b} = \frac {\alpha_{[a_1, \ldots, a_k]}}{k!}\left(\ders k f(0 \mid e_{a_1}, \ldots e_{a_k} ) \right)_b. $$ Recall from Lemma~\ref{lemma:additivity} that for any $\vec u \in  \boule {(\coneone_x^k)}$, the function $f \in \cstab(\functor \pcsone, \functor \pcstwo) \mapsto \ders k f(x \mid \vec u) \in \functor \pcstwo$ is linear and Scott-continuous. As a consequence, $\Xi^{\pcsone, \pcstwo}$ too is linear and Scott-continuous. \shortv{Moreover, it also preserves the unit ball (see the proof in the long version), hence is a morphism in $\cstab(\functor \pcsone \Rightarrow \functor \pcstwo, \functor (\pcsone \Rightarrow \pcstwo))$.}\longv{

  To know that $\Xi^{\pcsone, \pcstwo}$ is stable, we have still to show that it is bounded: we are actually going to show that it preserves the norm. Indeed, for every $f \in \boule (\functor \pcsone \Rightarrow \functor \pcstwo)$, we see using the definition of the norm on a cone obtained from a PCS (see Definition~\ref{def:pcstocone}), that:
  \begin{equation}\label{eq:nonexpansivexi1}
\norm{\functor(\pcsone \Rightarrow \pcstwo)}{\Xi^{\pcsone, \pcstwo}(f)} = \inf \{\frac 1 r \mid r \cdot \Xi^{\pcsone, \pcstwo}(f) \in \prog(\pcsone \Rightarrow \pcstwo)\}  
  \end{equation}
  It was shown in ~\cite{pcsaamohpc} that:
  \begin{equation}\label{eq:nonexpansivexi}
 r \cdot \Xi^{\pcsone, \pcstwo}(f) \in \prog(\pcsone \Rightarrow \pcstwo) \,\Leftrightarrow\, \forall x \in \prog \pcsone, \, (r \cdot \Xi^{\pcsone, \pcstwo}(f)) \cdot x^! \in \prog \pcstwo.
 \end{equation}
  We see that $(r \cdot \Xi^{\pcsone, \pcstwo}(f)) \cdot x^! = r \cdot f(x)$ since $\Xi^{\pcsone, \pcstwo}$ has been defined as the inverse of $\Upsilon^{\pcsone, \pcstwo}$. It means that we can rewrite~\eqref{eq:nonexpansivexi} as:
   \begin{equation}\label{eq:nonexpansivexibis}
 r \cdot \Xi^{\pcsone, \pcstwo}(f) \in \prog(\pcsone \Rightarrow \pcstwo) \,\Leftrightarrow\, \forall x \in \prog \pcsone, \, r\cdot f(x) \in \prog \pcstwo.
   \end{equation}
Since for every PCS $\pcsthree$, it holds that $\prog \pcsthree = \boule \functor \pcsthree$, we can now use~\eqref{eq:nonexpansivexibis} to rewrite~\eqref{eq:nonexpansivexi1} as:
   \begin{equation}\label{eq:nonexpansivexi3}
\norm{\functor(\pcsone \Rightarrow \pcstwo)}{\Xi^{\pcsone, \pcstwo}(f)} = \inf \{\frac 1 r \mid \forall x \in \boule \functor \pcsone, r \cdot f(x) \in \boule \functor \pcstwo\}  
  \end{equation}
   Looking now at the definition of the norm in the cone $\functor \pcsone \Rightarrow \functor \pcstwo$, we can complete the proof using~\eqref{eq:nonexpansivexi3} and the homogeneity of the norm. Indeed:
   \begin{align*}
     \norm{\functor(\pcsone \Rightarrow \pcstwo)}{\Xi^{\pcsone, \pcstwo}(f)} &= \inf \{\frac 1 r \mid \norm{\functor \pcsone \Rightarrow \functor \pcstwo}{ r \cdot f} \leq 1\} \\
     & = \inf \{\frac 1 r \mid r \cdot  \norm{\functor \pcsone \Rightarrow \functor \pcstwo}{ f} \leq 1\} \\
     &= \norm {\functor \pcsone \Rightarrow \functor \pcstwo}{ f}
  \end{align*}
}
  \end{proof}

As a direct consequence of Lemma~\ref{lemma:fpresprod} and Lemma~\ref{lemma:fpresarr}, we can state the following theorem:

\begin{theorem}\label{th:fullness}
$\functor$ is full and faithful, and it respects the cartesian closed structures.
  \end{theorem}

\section{Adding Measurability Requirements}
In~\cite{pse}, the authors developed a sound and adequate model of $\pcf_{\text{sample}}$ based on stable functions. However, as explained in more details in~\cite{pse}, they need to add to their morphisms some \emph{measurability requirements}, both on cones and on functions between them, since the denotational semantics of the $\texttt{let}(x, \termone, \termtwo)$ construct uses an integral, to model the fact that $\termone$ is evaluated before being passed as argument to $\termtwo$.

We call measurable functions $\RR^n \rightarrow \RR^k$ the functions measurable when both $\RR^n$ and $\RR^k$ are endowed with the Borel $\Sigma$-algebra associated with the standard topology of $\RR$. The relevant properties of the class of measurable functions $\RR^n \rightarrow \RR^k$ is that they are closed by arithmetic operations, composition, and pointwise limit, see for example Chapter 21 of~\cite{schechter1996handbook}.
\subsection{The category $\cstabm$}

$\cstabm$ is built as a \emph{refinement} of the category $\cstab$.
The objects of $\cstabm$ are going to be complete cones, endowed with a family of \emph{measurability tests}.

If $\coneone$ is a complete cone, we denote by $\coneone'$ the set of linear and Scott-continuous functions $\coneone \rightarrow \RR_+$. 
\begin{definition}\label{def:meascone}
  A \emph{measurable cone} (MC) is a pair consisting of a cone $\coneone$, and a collection of \emph{measurability tests} $ (\meastests \coneone n )_{n \in \NN})$, where for every $n$, $\meastests \coneone n \subseteq {{\coneone'}^{\RR^n }}$, such that:
  \begin{itemize}
  \item for every $n \in \NN$, $0 \in  \meastests \coneone n$;
  \item for every $n,p \in \NN$, if $l \in \meastests \coneone n$, and $h : \RR^p \rightarrow \RR^n$ is a measurable function, then $l \circ h \in \meastests \coneone p$;
    \item for any $l \in \meastests \coneone n$, and $x \in \coneone$, the function $u \in \RR^n \mapsto l(u)(x) \in \RR$ is measurable.
    \end{itemize}
\end{definition}
\begin{example}[from~\cite{pse}]
  Let $X$ be a measurable space.
 We endow the cone of finite measures $\meas X$  with the family $\meastests{X} {}$ of measurable tests defined as:
  $$\meastests {X}n = \{\epsilon_U \mid U \in \Sigma_X \}  \quad \text{where}\quad \epsilon_U(\vec r)(\mu) = \mu(U),$$
  where $\Sigma_X$ is the set of all measurable subsets of $X$. Observe that in this case, the measurable tests correspond to the measurable sets. In the following, we will denote $\measm X$ the measurable cone $(\meas X, (\meastests X n)_{n \in \NN})$.
  \end{example}
We define now \emph{measurable paths}, which are meant to be the \emph{admissible} ways to send $\RR^n$ into a MC $\coneone$.
\begin{definition}[Measurable Paths]
Let be $(\coneone, (\meastests \coneone n)_{n \in \NN})$ a measurable cone. A \emph{measurable path of arity $n$ at $\coneone$} is a function $\pathone:\RR^n \rightarrow \coneone$, such that $\gamma(\RR^n)$ is bounded in $\coneone$, and for every $k \in \NN$, for every $l \in \meastests \coneone k$, the function $(\vec r, \vec s) \in \RR^{k+n} \mapsto l(\vec r) (\pathone(\vec s)) \in \Rp{}$ is a measurable function.
\end{definition}
We denote $\pathes \coneone n$ the set of measurable paths of arity $n$ for the MC $\coneone$ . When a measurable path $\pathone$ verify  $\pathone(\RR^n) \subseteq \boule \coneone$ , we say it is \emph{unitary}. 
Using measurable paths, the authors of ~\cite{pse} add \emph{measurability requirements} to their definition of stable functions.
\begin{definition}\label{def:measstabfunctions}
  Let be $\coneone, \conetwo$ two MCs. 
A stable function $f:\boule \coneone \rightarrow \conetwo$ is \emph{measurable} if for all unitary $\pathone \in \pathes \coneone n$,  $f \circ \gamma \in \pathes \conetwo n$.
  \end{definition}
The category $\cstabm$ is therefore the category whose objects are MCs, and whose morphisms are measurable stable functions between MCs.
\begin{example}
  Recall the function $\semm {\texttt{real}}$ defined in Section~\ref{sect:overview}:
  $$\semm {\texttt{real}}: \mu \in \meas \NN \mapsto (U \in \Sigma_\RR \mapsto \sum_{n \in \NN \cap U} \mu(n) ) \in \meas {\RR}.$$
  We can see that $\semm{\texttt{real}}$ is a measurable function from $\measm \NN$ into $\measm \RR$. Moreover it is linear and Scott-continuous, and norm-preserving, which makes it a morphism in $\cstabm$. In the same way, taking  $\measm \NN$ as the denotational semantics of type $N$, we could complete the denotationnal semantics given in~\cite{pse} for $\pcf_{\texttt{sample}}$ in $\cstabm$ into a denotational semantics for $\pcf_{\oplus, \texttt{sample}}$.

  Observe that $\semm{\texttt{real}}$ would not be measurable, if we endowed $\meas \NN$ with for instance $\{0\}$ as measurability tests instead of $\meastests \NN{}$: indeed in that case, every $\gamma: \RR^n \rightarrow \meas \NN$ would be a measurability path. As a consequence, to be a measurable function, $\semm {\texttt{real}}$ should verify: for every arbitrary function $\gamma: \RR^n \rightarrow \meas \NN$,
$\semm{\texttt{real}} \circ \gamma$ is a measurable path on $\measm \RR$. However, we can see this is not the case,
for instance by considering $\gamma$ of the form $\gamma (s)=\alpha(s)\cdot \dirac{0}$, where $\alpha: \RR \rightarrow \Rp{}$ is not Borel measurable.
  
  \end{example}

In~\cite{pse}, the cartesian
closed structure of $\cstabm$ is derived from that of $\cstab$ by endowing its exponentials and products with the measurability tests presented in Figure~\ref{fig:cccstabm}.
\begin{figure}[!h]
\begin{center}
  \fbox{
  \begin{minipage}{0.45 \textwidth}
    \footnotesize
    \begin{center}
      $$\meastests {\prod_{i \in I} \coneonem_i} n = \{
      \bigoplus_{i \in I} l_i \mid \forall i \in I, l_i \in \meastests {\coneonem_i} n \} \quad \text{ with } I \text{ finite set.}$$
      $$\meastests {\coneonem \Rightarrow_m \conetwom} n = \{\trmeas \gamma m  \mid \gamma \in \pathes{\coneonem} n, m \in \meastests {\conetwom} n \}, $$
      with
$(\bigoplus_{i \in I} l_i(\vec r))((x_i)_{i \in I}) = \sum_{i \in I} l_i(\vec r)(x_i) \in \Rp{};$
     { $\text{ and }(\trmeas \gamma m) (\vec r)(f) = m(\vec r)(f(\gamma(\vec r))) .$}
    \end{center}
  \end{minipage}}
  \end{center}
  \caption{Cartesian Closed structure of $\cstabm$.}\label{fig:cccstabm}
  \end{figure}

\subsection{$\pcoh_!$ is a full subcategory of $\cstabm$}
We want now to convert the functor $\functor: \pcoh_! \rightarrow \cstab$ into a functor $\functorm: \pcoh_! \rightarrow \cstabm$. To build $\functorm$, we are going to endow each $\functor \pcsone$ with measurability tests, in such a way that $\functor(f)$ will be a \emph{measurable} stable function for any morphism $f \in \pcoh_!$.

Observe that this requirement does not determine uniquely the choice of measurability tests. For instance, it would be verified if we choose $\{0\}$ as measurability tests for every $\functor \pcsone$. However, as explained in Section~\ref{sect:overview},
we want also $\functorm (\NN^{\pcoh})$ to be isomorphic to $\semm{N}$: we would like to be able to inject any discrete distribution on $\NN$ into a distribution on $\RR$.
%
A natural way to ensure this is to use the discrete structure of the web to give the following definition of the MC arising from a PCS.
\begin{definition}
  For any $\pcsone \in \pcoh$, we define $\pcstoconem \pcsone$ as the measurable cone $\pcstocone \pcsone$ endowed with the family ${\meastests \pcsone n}_{n \in \NN}$ of measurability tests defined as
   $\meastests \pcsone n = \{0 \} \cup\{ \epsilon_a \mid a \in \web \pcsone\}$, where $\epsilon_a(\vec r,x) = x_a.$
  \end{definition}
We see that the $\epsilon_a$ are indeed linear (i.e commuting with linear combinations), and moreover Scott-continuous: hence they are indeed element of $\pcstocone {\pcsone}'$. It is easy to verify that the other conditions are verified, and so $\pcstoconem \pcsone$ is indeed a MC.

\begin{lemma}\label{lemma:mpforpcs}
  Let be $\pcsone$ a PCS. Then $\pathes {\pcstoconem \pcsone} n$ is the set of those $\gamma: \RR^n \rightarrow \pcstocone \pcsone$ such that:
  \begin{itemize}
  \item $\exists \lambda \in \RR, \gamma(\RR^n) \subseteq \lambda \boule \pcstocone \pcsone$
  \item $\forall a \in \web \pcsone$, $\gamma_a: \vec r \in \RR^n \mapsto \gamma(\vec r)_a \in \Rp{}$ is measurable.
    \end{itemize}
\end{lemma}
Two MCs with the same underlying cone, but different measurability tests may be isomorphic in $\cstab$: it is enough for them to have the same \emph{measurable paths}. It is what happens in the example below, where we consider $\pcstoconem{\NN^\pcoh}$ and $\measm \NN$. It is actually also what happens at higher-order types, as we will explain in Section~\ref{subsect:fmcc}.
\begin{example} 
The two measurable cones $\pcstoconem {\NN^\pcoh}$ and $\measm \NN$ have the same underlying cone, but they do not have the same measurable tests. Indeed:
$$
  \meastests {\pcstoconem{\NN^\pcoh}}n = \{\epsilon_n \mid n \in \NN \}; \quad \meastests{\measm \NN} n = \{ \epsilon_U \mid U \subseteq \NN\}.$$
  But we can prove that they have the same measurable paths. It is immediate that $\pathes{\measm \NN} n \subseteq \pathes {\pcstoconem{\NN^\pcoh}} n$, since $\meastests {\pcstoconem{\NN^\pcoh}}n $ is a subset of $ \meastests{\measm \NN} n $.  We detail now the proof of
  the reverse inclusion. Let  $\gamma \in \pathes  {\pcstoconem{\NN^\pcoh}}n$. We have to show: for every $U \subseteq \NN$, the function $$\vec r, \vec s \in \RR^{k+n} \mapsto \epsilon_U(\vec r)(\gamma(\vec s)) \qquad \text{is Borel measurable.}$$
  The key observation now is that $\epsilon_U(\vec r)(\gamma(\vec s) = \sum_{m \in U} \epsilon_m(\vec r)(\gamma(\vec s))$.
  Since $\gamma \in \pathes  {\pcstoconem{\NN^\pcoh}}n$ it holds that for every $m \in \NN$, the function $((\vec r,\vec s) \in \RR^{k+n} \mapsto \epsilon_m(\vec r,\gamma(\vec s)) \in \Rp{})$ is Borel measurable. Since the class of Borel measurable functions are closed by finite sum and pointwise limit, it leads to the result.


  \end{example}

\begin{lemma}\label{lemma:functormorph1}
  Let  $\pcsone, \pcstwo$ be two PCSs, and
  $f \in \pcoh_!(\pcsone, \pcstwo)$. Then $\functor f$ is measurable from $\pcstoconem \pcsone$ into $\pcstoconem \pcstwo$. 
\end{lemma}
\longv{\begin{proof}
  We have to show that $\functor f$ preserves measurable paths. Let $\gamma$ a unitary path in $\pathes {\pcstoconem \pcsone} n$: our goal is to show that $f \circ \gamma \in \pathes{\pcstoconem \pcstwo} n$. Recall that Lemma~\ref{lemma:mpforpcs} gives us a characterization of $\pathes{\pcstoconem \pcstwo} n$. Since $\gamma$ and $\functor f$ are bounded, we see immediately that $\functor f \circ \gamma$ is bounded. Let $b$ be in $\web \pcstwo$. We see that:
  $$(\functor f\circ \gamma)_b (\vec r) = \sum_{\mu \in \mfin {\web \pcsone}} f_{\mu,b}\cdot \prod_{a \in \text{Supp}(\mu)} {\gamma_a(\vec r)}^{\mu(a)} .$$
  Since $\gamma \in \pathes{\pcstoconem \pcsone} n$, it holds that $\gamma_a : \RR^n \rightarrow \Rp{}$ is measurable for all $a \in \web \pcsone$. We conclude by using the fact that the class of measurable functions $\RR^n \rightarrow \Rp{}$ is closed under multiplication, finite sums and limit of non-decreasing sequences: it tells us that $\vec r \in \RR^n \mapsto (\functor f\circ \gamma)_b (\vec r) \in \Rp{}  $ is measurable, and the result folds. 
\end{proof}}
\shortv{
    The proof can be found in the long version. It uses the characterization of $\pathes{\pcstoconem \pcsone} n$ given in Lemma~\ref{lemma:mpforpcs}.
}

\begin{theorem}
  The functor $\functorm : \pcoh_! \rightarrow \cstabm$ defined as $\functorm \pcsone = \pcstoconem \pcsone$, and $\functorm f = \functor f$, is full and faithful.
  \end{theorem}
\begin{proof}
  Observe that we can decompose the functor $\functor$ as $\functor = \forgetful \circ \functorm$, where $\forgetful$ is the forgetful functor from $\cstabm$ to $\cstab$. We know from Section~\ref{sect:proofconsext} that $\functor$ is full and faithful. Moreover, it holds that $\forgetful$ is faithful. From there, we are able to deduce the result\shortv{ (see the long version for more details).}\longv{:
    \begin{itemize}
    \item $\functorm$ is faithful: it is implied by the fact that $\functor$ is faithful.
      \item $\functorm$ is full: indeed suppose that it is not the case: then there exist two PCSs $\pcsone, \pcstwo$, and  $f \in \cstabm(\functorm \pcsone, \functorm \pcstwo$, such that $f$ is not in the image by $\functorm$ of $\pcoh(\pcsone, \pcstwo)$. Then we consider $g \in \cstab(\functor \pcsone, \functor \pcstwo)$ defined by $g = \forgetful (f)$. Since $\forgetful$ is faithful, there is no other $f'$ such that $g = \forgetful {(f')}$: it means that $g$ is not in the image by $\forgetful \circ \functorm$ of $\pcoh(\pcsone, \pcstwo)$. But since $\functor = \forgetful \circ \functorm$ is full, we have a contradiction. 
      \end{itemize}}
  \end{proof}
\subsection{$\functorm$ is cartesian closed.}\label{subsect:fmcc}
We want now to show that just as $\functor$, $\functorm$ is cartesian closed. Since the forgetful functor from $\cstabm$ to $\cstab$ is cartesian closed, we see that we have only to show that the $\cstab$-morphisms $\Psi^{\fone}$, $\Theta^{\fone}$, $\Upsilon^{\pcsone, \pcstwo}$ and $\Xi^{\pcsone, \pcstwo}$ defined in Lemmas~\ref{lemma:fpresprod} and Lemma~\ref{lemma:fpresarr} proofs, are also morphisms in $\cstabm$.

\begin{lemma}\label{lemma:auxmeaspath}
  Let $\pcsone$ be a PCS, $\coneonem$ any measurable cone, and $f \in \cstab(\forgetful ({\coneonem}), \functor \pcsone)$. We suppose that for every unitary $\gamma \in \pathes \coneonem n$:
  $$\forall a \in \web \pcsone, \, (f\circ \gamma)_a : \RR^n \rightarrow \Rp{} \text{is (Borel) measurable.}$$
  Then it holds that $f \in \cstabm(\coneonem, \functorm \pcsone)$.
\end{lemma}
\begin{proof}
 Since we already know that $f$ is a morphism in $\cstab$, hence we have only to show that it preserves measurable paths. Let $\gamma$ be unitary in $\pathes \coneonem n$. We are going to use Lemma~\ref{lemma:mpforpcs} to show  that $f \circ \gamma$ is a measurable path for $\functorm \pcsone$. The second condition in Lemma~\ref{lemma:mpforpcs} holds by hypothesis. The first condition also holds: since both $f$ and $\gamma$ are bounded,  $f \circ \gamma$ is bounded too. So  Lemma~\ref{lemma:mpforpcs} tells us that $f \circ \gamma \in \pathes {\functorm \pcsone} n$.
  \end{proof}


\begin{lemma}\label{lemma:functormcartesian}
  For all $\fone = (\pcsone_i)_{i \in I}$ a finite family of PCSs,
  \begin{align*}
    \Psi^{\fone} & \in \cstabm( {\functorm{(\prod_{i \in I}^{\pcoh_!} \pcsone_i)} }, \prod_{i \in I}^{\cstabm} \functorm{\pcsone_i} )\\
\text{and} \quad    \Theta^{\fone} & \in \cstabm( \prod_{i \in I}^{\cstabm} \functorm{\pcsone_i}, {\functorm{(\prod_{i \in I}^{\pcoh_!} \pcsone_i)} } ).
    \end{align*}
\end{lemma}
\begin{proof}
  \begin{itemize}
  \item Recall that $\Psi^{\fone}$ is defined canonically in Equation~\eqref{eq:morphismcartproduct} as $\Psi^{\fone} = \langle \functor (\pi_i) \mid i \in I \rangle$, where $\langle \cdot \rangle$ is the cartesian product on morphisms in $\cstab$. Since the cartesian product on morphisms in $\cstabm$ is the same as the one in $\cstab$ (see~\cite{pse}), and moreover $\functor(\pi_i) = \functorm(\pi_i)$, we see that $\Psi^{\fone}$ is also a morphism of $\cstabm$.
  \item Using Lemma~\ref{lemma:auxmeaspath}, we see that it is enough to show that for all $\gamma$ in $\pathes {\prod_{i \in I}^{\cstabm} \functorm{\pcsone_i}} n$,
    for all $(i,a_i) \in \web {\prod_{i \in I}^{\pcoh_!} \pcsone_i}  $, $(\Theta^\fone\circ \gamma)_{(i,a_i)}$ is measurable.
By looking at the definition of $\Theta^\fone$,
  we see that $(\Theta^\fone\circ \gamma)_{(i,a_i)} (\vec r) = (\gamma(\vec r)_i)_a$. We see now that we can construct a measurability test $m \in\meastests{{\prod_{i \in I}^{\cstabm} \functorm{\pcsone_i}}} 0 $ such that $ (\gamma(\vec r)_i)_a = m(\cdot)(\gamma(\vec r))$: it is enough to take $m = \oplus_{j \in I} l_j$, with $l_j = 0$ if $j \neq i$, and $l_i = \epsilon_{a_i}$. Since $\gamma$ is a measurability test, it means that $\vec r \in \RR^n \mapsto  m(\cdot)(\gamma(\vec r)) \in \Rp{}$ is measurable, and so the result holds.
  \end{itemize}
\end{proof}
Lemma~\ref{lemma:functormcartesian} allows us to see that $\functorm$ is a cartesian functor. We want now to show that it also respects the $\Rightarrow$ construct. First, we show that the $\cstab$ morphism $\Upsilon^{\pcsone, \pcstwo}$ is also a morphism in $\cstabm$.
\begin{lemma}
 For all $\pcsone$, $\pcstwo$ PCSs,
  $$\Upsilon^{\pcsone, \pcstwo} \in \cstabm( \functorm (\pcsone \Rightarrow \pcstwo), \functorm \pcsone \Rightarrow \functorm \pcstwo)$$
\end{lemma}
\begin{proof}
Recall that $\Upsilon^{\pcsone, \pcstwo}$ is defined using currying in $\cstab$, $\Theta^{\pcsone \Rightarrow \pcstwo, \pcsone}$, and the eval morphism in $\cstab$. Since currying and structural morphisms are the same in $\cstabm$ as in $\cstab$, and moreover we have shown in Lemma~\ref{lemma:functormcartesian} that $\Theta^{\pcsone \Rightarrow \pcstwo, \pcsone}$ is a morphism in $\cstab$, we have the result.
\end{proof}
We show now that $\Xi^{\pcsone, \pcstwo}$ is also a $\cstabm$ morphism, by
using Lemma~\ref{lemma:auxmeaspath}.  To do that, we need to show that  $(\Xi^{\pcsone, \pcstwo} \circ \gamma)_{\mu,b}$ is Borel measurable for every $(\mu, b) \in \web{\pcsone \Rightarrow \pcstwo}$. Our proof strategy is the following: first we show that it can be written as a higher-order partial derivative of a (Borel) measurable function, and then we show that under some conditions on the domain, the partial derivative of a Borel measurable function is again Borel measurable.
\begin{lemma}\label{lemma:derivaux1}
  Let be $\mu \in \mfin {\web \pcsone}$, $b \in \web \pcstwo$. Let be $\{a_1, \ldots, a_p\}$ the support of $\mu$. Then there exists $\delta \in \pathes{\functorm \pcsone} p$, and $\alpha_\mu >0$, such that for every $f \in \prog{(\pcsone \Rightarrow \pcstwo)}$, the function:
  $$\psi_{\delta}^f: \vec t \in \RR^p \mapsto (\trmeas \delta {\epsilon_b})(\vec t)(\mapm f) \in \Rp{} $$
  verify: $\exists c>0$, such that the partial derivative  $\frac {\partial {(\psi_\delta^f)}^{\card \mu}}{\partial {t_1}^{\mu(a_1)} \ldots \partial{t_p}^{\mu(a_p)}}$ exists on $[0,c]^p \subseteq \RR^p$, and moreover its value in $\vec 0$ is $\alpha_\mu \cdot f_{\mu,b}$.
  \end{lemma}
\begin{proof}
  We take
$$ \delta:  \vec t \in \RR^p \mapsto \begin{cases}
    \sum_{1 \leq i\leq m} t_i \cdot e_{a_i} \text{ if } t_i \geq 0 \forall i \text{ and }\sum_{1 \leq i \leq m} t_i \leq 1; \\
    0 \text{ otherwise.}
  \end{cases} $$
  Using Lemma~\ref{lemma:mpforpcs}, we see that indeed $\delta \in \pathes {\functorm \pcsone} p$. 
  Observe that $\psi_{\delta}^f(\vec t) = \sum_{\nu \mid \supp \nu \subseteq{\{a_1, \ldots,a_n\}}} f_{\nu,b} \cdot \vec t^\nu.$
  From there, by using therorems of real analysis for normally convergent series of functions, we can deduce the result (the complete proof may be found in the long version). 
  \end{proof}


\longv{\begin{proof}
    Since $\gamma$ is a measurable path, we know that for every $p \in \NN$, and $l \in \pathes {\functorm \pcsone} p$, $\trmeas {\epsilon_b}{l}$ is a measurability test on ${\functorm \pcsone \Rightarrow \functorm \pcstwo}$, and therefore:
\begin{equation}\label{eq:fmarroweqaux1}
  (\vec r, \vec u) \in \RR^{p+n} \mapsto (\trmeas {\epsilon_b} l (\vec r))(\gamma(\vec u)) \in \Rp{} \text{ is measurable.}
  \end{equation}
We are going to apply~\eqref{eq:fmarroweqaux1} to a particular measurable path on ${\functorm \pcsone}$. Let $p$ be the cardinality of $\text{Supp}(\mu)$, and $\{a_1, \ldots, a_p\} = \text{Supp}(\mu)$. We define $l^{\mu}: \RR^{p} \rightarrow \functorm \pcsone$ as:
 $$ l^{\mu}:  \vec r \in \RR^p \mapsto \begin{cases}
    \sum_{1 \leq i\leq m} r_i \cdot e_{a_i} \text{ if } r_i \geq 0 \forall i \text{ and }\sum_{1 \leq i \leq m} r_i \leq 1; \\
    0 \text{ otherwise.}
\end{cases} $$
We see that $l^{\mu}(\RR^p)$ is bounded in $\functorm \pcsone$, and moreover for every $a \in \web \pcsone$, the function $\vec r \in \RR^p \mapsto l^{\mu}(\vec r)_a \in \RR^+$ is measurable. Using the characterization of $\pathes {\functorm \pcsone} p$ in Lemma~\ref{lemma:mpforpcs}, we see that $l^{\mu}$ is  in $\pathes {\functorm \pcsone} p$. Thus we can apply~\eqref{eq:fmarroweqaux1} with $l = l^{\mu}$. Observe that
    $$(\trmeas {\epsilon_b} l^{\mu} (\vec r))(\gamma(\vec u)) =  \left(\gamma(\vec u)(l^\mu(\vec r))\right)_b.$$  Therefore~\eqref{eq:fmarroweqaux1} tells us that $\phi^{\mu,b}: \RR^{p+n} \rightarrow \Rp{}$ is measurable,
with $\phi^{\mu,b}$ defined as
$\phi^{\mu,b}: (\vec r, \vec u) \in \RR^{p+n} \mapsto \gamma(\vec u)(l^\mu(\vec r))_b \in \RR_+.$
We define $J \subseteq \RR^p$ as $[0,\frac 1 p[^p$.
We are going to look at the restriction of the function $\phi^{\mu,b}$ to $J \times \RR^n$: indeed we are going to show that $\phi^{\mu,b}$ has partial derivatives on that interval.
We define $\psi^{\mu,b} : J \times \RR^n \rightarrow \Rp{}$ as the restriction of $\phi^{\mu,b}$ to $J \times \RR^n$.
Since $\phi^{\mu,b}$ is a measurable function, and $J \times \RR^n$ a measurable subset of $\RR^{p+n}$, $\psi^{\mu,b}$ also is measurable. 

Lemma~\ref{lemma:partder1} below (which is proved in the long version) is key: it says that we can recover the coefficients of the power series $\psi^{\mu,b}$ by looking at its partial derivatives. We will then show that we can do it in a \emph{measurable} way. 
\begin{lemma}\label{lemma:partder1}
  For every multiset $\nu \in \mfin{\{1, \ldots, p\}}$, there exists an interval $K$ of the form $[0,c]^p$ such that the partial derivative 
  $\partial^{\nu} \psi^{\mu,b} = \frac {\partial {(\psi^{\mu,b}_{\mid K \times \RR^n})}^{\card \nu}}{\partial {r_1}^{\nu(1)} \ldots \partial{r_p}^{\nu(p)}}: K \times \RR^n \rightarrow \Rp{}$ exists, and moreover:
   $$\partial^{\nu}\psi^{\mu,b}(\vec 0, \vec u) =  {\Xi^{\pcsone, \pcstwo}(\gamma(\vec u))}_{\nu,b}  \cdot \prod_{1 \leq i \leq p} {\nu(i)!} $$
\end{lemma}
\longv{
\begin{proof}
 Since $l^{\mu}(\vec r) = \sum_{1 \leq i \leq p} r_i \cdot e_i$ for $\vec r \in J$, we see that:
$$
 \phi^{\mu,b}(\vec r, \vec u) = \sum_{\nu \in \mfin{\web \pcsone}} {\Xi^{\pcsone, \pcstwo}({\gamma(\vec u)})}_{\nu,b} \cdot \vec r^{\nu} \in \Rp{}. $$
 For a fixed $\vec s$, we can see it as a generalization of entire series in real analysis. There are well-known results about the differentiation of such series: for instance, a uniformly convergent entire series is differentiable on its (open) domain of convergence.
Here, we are going to show the counterpart of some properties on entire series, on what we call \emph{multisets series of $p$ real variables}: those are the series of the form $$S(\vec r) = \sum_{\nu \in \mfin{1, \ldots,p}} a_\nu \cdot \vec r^\nu \quad \text{where} \quad \vec r \in \RR^p.$$  
 First, we observe that for each $\vec r$, we can look at $S(\vec r)$ as an infinite sum over natural numbers:
 $$S(\vec r) = \sum_{n \in \NN} (\sum_{\mu \in \mfin{1, \ldots,m} \mid \card \mu = n} a_\mu \cdot\vec r^\mu) .$$
We recall here a classical result of real analysis on power series, that we will use in the following.
\begin{lemma}\label{derserieonev}[Derivation of a series]
Let $f_n : I \rightarrow \RR$ be a sequence of functions from a bounded interval $I$. We suppose that $f(x) = \sum_{n \in \NN} f_n(x)$ is convergent for every $x \in I$, and moreover for each $n \in \NN$, $f_n$ is derivable and $\sum_{n \in \NN} f_n'$ is uniformly convergent on $I$. Then $f$ is derivable, and moreover $f' = \sum_{n \in \NN} f'_n$.
  \end{lemma}

\begin{lemma}\label{lemma:derpart1}
  Let $p \in \NN$, and $S(\vec r) = \sum_{\nu \in \mfin{\{1, \ldots,p\}}} a_\nu \cdot \vec r^\nu $ with non-negative coefficients $a_\mu$, such that $S$ is convergent on an interval $I=[ -c ,c ]^p$, with $c >0$.

  Then there exists $0 <b < c$, such that the function
     $g:\vec r \in ]-b, b[^p \mapsto S(\vec r) \in \RR $ is partially derivable in each of the $r_i$ variables, and moreover:
     $$\frac {\partial g}{\partial {r_i}}(\vec r) = \sum_{\nu \in \mfin{\{1, \ldots,p\}} \mid i \in \nu} a_\nu \cdot \vec r^{\nu - [i]}  \cdot \nu(i).$$
\end{lemma}

\begin{proof}
  We take $b = \frac c 2$, and set $J=]-b, b[$. To simplify the notations, we suppose here that $i = 1$, but the proof is the same in other cases.
   We want to show that for any fixed $\vec u \in J^{p-1}$, the function
     $ h_{\vec u}:  r \mapsto g(r, \vec u)$ is derivable on $J$.
     Let us fix $\vec u \in ]-b, b[^{p-1}$. We are going to use Lemma~\ref{derserieonev} on $h_{\vec u}$.
  We see that $h_{\vec u}( r) = \sum_{n \in \NN} h_n(r)$, where $$h_n(r) = (\sum_{\nu \in \mfin{2, \ldots,p} } a_{\nu+[1^n]} \cdot \vec u ^\nu ) \cdot r^n,$$
  where $[1^n]$ is the multiset consisting of $n$ occurrences of $1$.
We see that for every $n \in \NN$, the function $h_n$ is derivable on $J$, and:
     $$h_n'(r) =(\sum_{\nu \in \mfin{2, \ldots,p}}  a_{\nu+[1^n]} \cdot \vec u ^\nu )\cdot n \cdot r^{n-1} .$$
      We see now that the series $\sum_{n \in \NN} h_n'$ is uniformly convergent on $J$: for every $r \in J$, it holds that:
     \begin{align*}
       \lvert h_n'(r) \rvert &\leq (\sum_{\nu \in \mfin{2, \ldots,p}}  a_{\nu+[1^n]}  \cdot \lvert \vec u \rvert ^\nu ) \cdot n \cdot \lvert r \rvert ^{n-1} \\
       &= (\sum_{\nu \in \mfin{2, \ldots,p}}  a_{\nu+[1^n]}  \cdot \lvert \vec u \rvert ^\nu \cdot c^n ) \cdot n \cdot \frac{1}{c}\cdot \left(\frac{\lvert r \rvert}{c}\right) ^{n-1} \\
       & = (\sum_{\eta \in \mfin{1, \ldots,p} \mid \eta(1) = n}  a_{\eta}  \cdot \lvert (c,\vec u) \rvert ^\eta ) \cdot n \cdot \frac{1}{c}\cdot \frac{\lvert r \rvert^{n-1}}{c^{n-1}}
     \end{align*}
     Since the series $S(\vec r) = \sum a_\mu \cdot \vec r ^\mu$ is convergent on $I$, and $(c, \lvert \vec u \rvert ) \in I$, it holds that there exists $M \geq 0$, with $\sum_{\nu \in \mfin{1, \ldots,p} \mid \nu(1) = n} a_{\nu}  \cdot \lvert (c,\vec u) \rvert ^\nu \leq M.$
 As a consequence, and since moreover for each $r \in J$, it holds that $\lvert r \rvert \leq b$, we can now write:
     \begin{equation}\label{eq:auxderivhn}
       \forall r \in J, \quad \lvert h_n'(r) \rvert \leq  M \cdot \frac {n} c \cdot  \left(\frac{\lvert r \rvert}{c}\right)^{n-1} \leq  M \cdot \frac {n} c \cdot  \left(\frac{b}{c}\right)^{n-1}
     \end{equation}
     Since $ b< c$, we know that the quantity in the right part of~\eqref{eq:auxderivhn} defines a convergent series. As a consequence, the series $\sum_{n \in \NN} h_n'$ is uniformly convergent on $J$, which means that we are able to apply Lemma~\ref{derserieonev}: we see that $\frac {\partial g}{\partial {r_1}}$ exists on $J^n$, and moreover:
\begin{align*}
  \frac {\partial g}{\partial {r_1}} ( \vec r) &= \sum_{n \in \NN} h_n'(r, (r_2, \ldots,r_p)) \\
  & = \sum_{n \in \NN}(\sum_{\nu \in \mfin{\{2, \ldots,p\}}}  a_{\nu+[1^n]} (r_2, \ldots,r_n) ^\nu )\cdot n \cdot r^{n-1}\\
  & =   \sum_{\nu \in \mfin{\{1, \ldots,p\}} \mid 1 \in \nu} a_\nu \cdot \vec r^{\nu-[1]} \cdot \nu(1) \quad \text{and the result holds.}
\end{align*}

\end{proof}
We iterate now Lemma~\ref{lemma:derpart1} in order to look at higher-order partial derivatives.
\begin{lemma}\label{lemma:derpart2}
  Let $S(\vec r) = \sum_{\nu \in \mfin{\{1, \ldots,p\}}} a_\nu\cdot (\vec r)^\nu $  with $a_\nu \geq 0$. We suppose $S$ convergent on $I=] - b ,b [^p$, with $b>0$.
     Then for every multiset $\nu \in \mfin{\{,1 \ldots, p\}}$, there exists $0 <b \leq a$, such that, when we define 
     $g:\vec r \in [-b, b]^p \mapsto S(\vec r) $, the partial higher-order derivative
     $\frac {\partial g^{\card \nu}}{\partial {r_1}^{\nu(1)} \ldots \partial{r_p}^{\nu(p)}}$ exists, and moreover:
     $$\frac {\partial g^{\card\nu}}{\partial {r_1}^{\nu(1)} \ldots \partial{r_p}^{\nu(p)}}(\vec r) = \sum_{\eta \in \mfin{\{1, \ldots,p\}} } a_{\nu+\eta} \cdot  {\vec r}^{\eta} \cdot \prod_{1 \leq i \leq p} \frac {(\eta+\nu)(i)!}{\eta(i) !}. $$
\end{lemma}
\begin{proof}
  The proof is by induction on $\card \nu$, and uses Lemma~\ref{lemma:derpart1}. It is clear that the result holds for $\nu= \emptyset$. Now, we suppose that it holds for every $\nu$ of cardinality $N$. Let $\kappa$ be a multiset of cardinality $N+1$, and we take $\nu$, and $i$ such that $\kappa = \nu + [i]$. By the induction hypothesis, there exists $c >0$, such that, when we define  $g:\vec r \in [-c, c]^p \mapsto S(\vec r) $, the partial derivative $\frac {\partial g^{\card\nu}}{\partial {r_1}^{\nu(1)} \ldots \partial{r_p}^{\nu(p)}}$ exists, and is equal to: $$T(\vec r)= \sum_{\eta \in \mfin{\{1, \ldots,p\}} } a_{\eta+\nu} \cdot  {\vec r}^{\eta} \cdot \prod_{1 \leq j \leq p} \frac {(\eta+\nu)(j)!}{\eta(j) !}. $$
  We see we can apply Lemma~\ref{lemma:derpart1} with $T$ as multiset series, and $I = [-c, c]^p$. It means that there exist $0 < d < c$, such that $\frac{\partial T}{\partial {r_i}}$ exists, and
  \begin{align*}
    \frac{\partial T}{\partial {r_i}}(\vec r) &=\sum_{\eta \in \mfin{\{1, \ldots,p\}} \mid i \in \eta} a_{\eta+\nu} \cdot \vec r^{\eta-[i]}\cdot \eta(i) \cdot \prod_{1 \leq j \leq p} \frac {(\eta+\nu)(j)!}{\eta(j) !} \\
    & = \sum_{\iota \in \mfin{\{1, \ldots,p\}}} a_{\iota+\kappa} \cdot  \vec r^\iota \cdot  {(\iota+[i])(i)} \cdot \prod_{1 \leq j \leq p} \frac {(\iota+\kappa)(j)!}{(\iota+[i])(j) !} \\
     & = \sum_{\iota \in \mfin{\{1, \ldots,p\}}} a_{\iota+\kappa} \cdot  \vec r^\iota \cdot  \prod_{1 \leq j \leq p} \frac {(\iota+\kappa)(j)!}{\iota(j) !}.
  \end{align*}
\end{proof}
We end the proof of Lemma~\ref{lemma:partder1} by using Lemma~\ref{lemma:derpart2} for each $\vec u \in \RR^n$ on the multiset series given by
$$S_{\vec u}(\vec r)= \sum_{\nu \in \mfin{\{1, \ldots,p\}}} (\Xi^{\pcsone, \pcstwo}{\gamma(\vec u)})_{\nu,b} \cdot \vec r^{\nu}.$$
We see that it is indeed absolutely convergent on $I = ]-\frac 1 p, \frac 1 p[^p$, using the fact that for $\vec r \in [0, \frac 1 p]^p$, $S_{\vec u}(\vec r) = \psi^{\mu,b}(\vec r, \vec u)$  for $\vec r \in \RR^p$.
\end{proof}}
  \end{proof}}
  \begin{lemma}
For every unitary $\gamma \in \pathes{\functorm \pcsone \Rightarrow \functorm \pcstwo} n$, and $(\mu,b) \in \web{\pcsone \Rightarrow \pcstwo}$, it holds that  $(\Xi^{\pcsone, \pcstwo} \circ \gamma)_{\mu,b}$ is Borel measurable for every $(\mu, b) \in \web{\pcsone \Rightarrow \pcstwo}$. 
    \end{lemma}
  \begin{proof}
    We take $\alpha_\mu, \delta, c$ as given by Lemma~\ref{lemma:derivaux1}. Since $\trmeas{\epsilon_b}{\delta}$ is a measurability tests for $\functorm \pcsone \Rightarrow \functorm \pcstwo$, we see that the function $$(\vec t, \vec s) \in \RR^{p+n} \mapsto \psi_\delta^{\gamma(\vec s )}(\vec t) = {(\trmeas {\epsilon_b}{\delta})(\vec t)(\gamma(\vec s))}  \in \Rp{}$$ is measurable. Since $K = [0,c[^p \times \RR^n$ is a measurable subset of $\RR^{p+n}$, the restriction (that we denote $\phi$) of this function to $K$ is measurable too.
        Moreover, observe that $\Xi^{\pcsone, \pcstwo} \circ \gamma (\vec s) \in \prog{(\pcsone \Rightarrow \pcstwo})$ and $\mapm{{\Xi^{\pcsone, \pcstwo} \circ \gamma(\vec s)}} = \gamma(\vec s) $. It tells us that we can apply Lemma~\ref{lemma:derivaux1}, and we see that:
\begin{equation}\label{eq:auxdermeaas}
  (\Xi^{\pcsone, \pcstwo}\circ \gamma)(\vec s)_{\mu,b} = \frac 1 {\alpha_\mu} \cdot  \frac {\partial {\phi}^{\card \mu}}{\partial {t_1}^{\mu(a_1)} \ldots \partial{t_p}^{\mu(a_p)}}(\vec 0, \vec s)
  \end{equation}
        (observe that this partial derivatives exists since it exists for $\psi_\delta^{\gamma(s)}$ for every fixed $\vec s$). 
        We are now going to show that every partial derivative of $\phi$, when it exists, is measurable too. It is based of the fact that the class of real-valued measurable functions is closed by addition, multiplication by a scalar and pointwise limit. Indeed, there exists a poitive sequence $(r_n)_{n \in \NN}$ in $[0,c[$ which tends towards $0$. As a consequence, the partial derivative of $\phi$ with respect to $t_1$ (for instance) may be written as: $ \frac {\partial \phi}{\partial t_1}(0, \vec t,\vec s) =  \lim_{n \rightarrow \infty} f_n( \vec t, \vec s)$, with
            $f_n(\vec t, \vec s) = \frac  {\phi((r_n, \vec t), \vec s) - \phi((0, \vec t), \vec s)}{r_n}$. It tells us that $((\vec t,\vec s) \in \RR^{p-1+n} \mapsto \frac {\partial \phi}{\partial t_1}(0, \vec t,\vec s)) $ is the pointwise limit of a sequence of measurable functions, hence is measurable. By iterating this reasonning, we see that it is also the case for higher-order partial derivatives (when they exists), and~\eqref{eq:auxdermeaas} allows us to conclude the proof. 
            
  \end{proof}
  
\begin{lemma}\label{lemma:fmarrowclosed}
  For all $\pcsone$, $\pcstwo$ PCSs,
  $$  \Xi^{\pcsone, \pcstwo} \in \cstabm(\functorm \pcsone \Rightarrow \functorm \pcstwo, \functorm (\pcsone \Rightarrow \pcstwo)).$$ 
\end{lemma}

\hide{\begin{proof}
Since the class of real-valued measurable functions is closed by addition, multiplication by a scalar and pointwise limit, and that $\psi^{\mu,b} : K \times \RR^p \rightarrow \Rp{}$ is measurable, it holds that the partial derivatives (when they exist) are measurable too. Indeed, observe that $\frac{\partial{\psi^{\mu,b}}}{\partial r_1}((r_1, \vec r), \vec u) = \lim_{n \rightarrow \infty} f_n((r_1, \vec r), \vec u)$, with
$$f_n((r_1, \vec r), \vec u) = n \cdot {f((r_1+ \frac 1 n, \vec r), \vec u) - f((r_1, \vec r), \vec u)}.$$
As a consequence, applying Lemma~\ref{lemma:partder1} with $\nu = \mu$ leads us to: 
$$\vec u \in \RR^n \mapsto \frac {\partial {(\psi^{\mu,b})}^{\card \mu}}{\partial {r_1}^{\mu(1)} \ldots \partial{r_p}^{\mu(p)}}(\vec 0, \vec u)  \text{ is measurable.}$$
Therefore (again by Lemma~\ref{lemma:partder1}), 
$(\vec u \in \RR^n \mapsto {\Xi^{\pcsone, \pcstwo}(\gamma(\vec u))}_{\mu,b}  \cdot \prod_{1 \leq i \leq p} {\mu(i)!} \text{ is measurable})$, and the result holds.
  \end{proof}
}
  
\begin{theorem}
$\functorm$ is a cartesian closed full and faithful functor.
\end{theorem}

\hide{\subsection{Programming Languages Considerations}
We want now to give a denotational semantics of $\pcf_\oplus$ in the category $\cstabm$. We want it to be as compatible as possible with the semantics of $\pcf_{\text{sample}}$ given in~\cite{pse}, in order to be able to see it as a fragment of the semantics of a language $\pcf_{\oplus, \text{sample}}$ containing both $\pcf_{\oplus}$ and $\pcf_{\text{sample}}$, which would be a language with continuous probabilities with an explicit discrete probabilistic fragment.

To ensure this compatibility, we take $\semm \NN = \meas \NN$, and the semantics of every type and typing context in $\pcf_\oplus$ is given inductively by $\semm \NN$, $\cstabm$ and $\times_m$.
We would like now to use the functor $\functorm$ in order to send $\sem{\pcf_\oplus}_{\pcoh_!}$ into $\cstabm$, in such a way that the semantics of types is as specified above. To do that, we need to have guarantees about the behavior of $\functorm$  regarding the cartesian structure. However, it is not true that $\functorm$ respects the cartesian structure of $\pcoh_!$: it comes from the fact that measurability tests for $\functorm{\pcsone} \Rightarrow_m \functorm {\pcstwo}$ are not the same as those in $\functorm{\pcsone \Rightarrow \pcstwo}$. But, it turns out that the \emph{measurable paths} of these two MCs are the same. It leads us to a weaker form of preservation of the cartesian structure, stated below in Theorem~\ref{th:cartclosedfunctorm}.
\begin{definition}
We say that two MCs $\coneone_m$ and $\conetwo_m$ are \emph{cone-isomorphic}, and we note $\coneone_m \equiv \conetwo_m$, if there exists an isomorphism $\phi$ in $\cstab$ between their underlying cones, and moreover $\pathes {\conetwo_m} n = \phi \circ \pathes {\coneone_m} n$.
\end{definition}
Observe that $\meas \NN \equiv \functor \NN$. The relevance of the equivalence relation comes from Theorem~\ref{th:cartclosedfunctorm} below:

\begin{theorem}\label{th:cartclosedfunctorm}
Let $\pcsone$ and $\pcstwo$  be two PCSs. Let $\coneone_m \equiv \functorm \pcsone$, and $\conetwo_m \equiv \functorm \pcstwo$. Then $(\coneone_m \Rightarrow_m \conetwo_m) \equiv \functorm{(\pcsone \Rightarrow \pcstwo)}$, and $(\coneone_m \times_m \conetwo_m) \equiv \functorm{(\pcsone \times \pcstwo)}$.
\end{theorem}
The proof, quite technical, of Theorem~\ref{th:cartclosedfunctorm} can be found in Appendix.


Now, we denote by $\discrm$ the subcategory of $\cstabm$ build inductively by $\coneone_\NN$, $\times_m$ and $\Rightarrow_m$, and $\discr$ the subcategory of $\pcoh_!$ obtained in the same way from $\NN$, $\Rightarrow$ and $\times$.
Observe that the $\pcoh_!$ semantics of $\pcf_{\oplus}$ actually live in $\discr$, and that we want to embed it into $\discrm$. We do this by using Theorem~\ref{th:cartclosedfunctorm} to convert $\functorm$ into a convenient functor $\discr \rightarrow \discrm$.

\begin{proposition}
  There exists a functor $\functord : \discr \rightarrow \discrm$, that respects the cartesian structure, and is full and faithful.
\end{proposition}
\begin{proof}
We first define the effect of $\functord$ on objects in $\discr$. For every $\typone$ a discrete type, we take $\functord{\sem \typone} = \semm \typone$. Now, we are going to define the effect of $\functord$ on morphisms. For every $\typone, \typtwo$ discrete type, we can see by iterating Theorem~\ref{th:cartclosedfunctorm} that $\functorm {\sem \typone} \equiv \semm \typone$, as well as $\functorm {\sem \typtwo} \equiv \semm \typtwo$. Let $\phi^\typone , \phi^\typtwo$ be the relevant cones isomorphisms between underlying cones. Let $f \in \pcoh_!(\sem \typone, \sem \typtwo)$: we take $\functorm f = \phi^{\typtwo} \circ \functor f \circ {\phi^{\typone}}^{-1}$. To show that it is full and faithful, we use the fact that $\functorm$ is so.
  \end{proof}

\begin{theorem}
$\functord(\sem{PCF_\oplus})$ is a fully abstract denotational model of $PCF_\oplus$ in $\cstabm$. 
\end{theorem}
\begin{proof}
  We know from~\cite{pcsaamohpc} that $\sem{{\pcf_\oplus}}$ is a fully abstract denotational model of $\pcf_\oplus$. 
Since $\functord$ is a full and faithful functor, it is also the case for $\functord(\sem{\pcf_\oplus})$.
  \end{proof}}
\begin{section}{Conclusion}
  Our full embedding of $\pcoh_!$ into $\cstab$ implies that every stable function $f$ from $\prog \pcsone$ to $\prog \pcstwo$ can be characterized by an element $\Xi^{\pcsone, \pcstwo} (f) \in \RR^{\mfin{\web{\pcsone}}\times \web \pcstwo}$, that has to be seen as a power series. It gives us a \emph{concrete representation} of stable functions on discrete cones, similar to the notion of trace introduced by Girard in~\cite{girard1986system} for stable functions on quantitative domains. There are well-known real analysis results on power series, as for instance the \emph{uniqueness theorem}\textemdash any power series which is null on an open subset has all its coefficients equal to $0$\textemdash on which is based the proof of full abstraction for $\pcf_\oplus$ in $\pcoh_!$~\cite{EPT15}.
  While we have not been able to extend such a concrete representation to cones which are not \emph{directed-complete}, as for instance the cone $\meas \RR \Rightarrow_m \meas \RR$, our result could hopefully be a first step in this direction. This kind of characterization could lead to a way towards a full abstraction result for the continuous language $\pcf_{\text{sample}}$ in $\cstabm$, and more generally gives us new tools to reason about continuous probabilistic programs. 
  \end{section}

\bibliographystyle{abbrv}
\bibliography{biblio}

\end{document}